\documentclass[letterpaper,11pt]{article}

\usepackage[margin=1in]{geometry}

\usepackage[T1]{fontenc}

\usepackage{amsmath, amsthm, amssymb, mathtools, bm, dsfont}
 
\usepackage[comma,longnamesfirst]{natbib}
\usepackage[dvips]{epsfig}         
\usepackage{graphicx}              
\usepackage[table,usenames,dvipsnames]{xcolor}  

\usepackage{dcolumn, hhline, multirow, booktabs}
\usepackage[flushleft]{threeparttable}
\usepackage{longtable}

\usepackage[percent]{overpic}
\usepackage{rotating}
\usepackage{wrapfig}
\usepackage{adjustbox}

\usepackage{enumerate}
\usepackage{enumitem}
\usepackage{paralist}  

\usepackage{afterpage}
\usepackage{placeins}
\usepackage{lscape}  

\usepackage{xr}
\usepackage[toc,page]{appendix}
\usepackage{chngcntr}

\usepackage[breaklinks]{hyperref}
\usepackage{breakurl}

\usepackage{caption}
\usepackage{subcaption}

\usepackage{algorithm}
\usepackage{algpseudocode}

\usepackage[thinc]{esdiff}

\usepackage{diagbox}
\usepackage{slashbox}

\usepackage{titlesec}    
\usepackage{sectsty}
\usepackage{setspace}

\usepackage{eurosym}
\usepackage{quiver}
\usepackage{rotating}

\usepackage{dsfont}

\DeclareUnicodeCharacter{0301}{\'{e}}

\renewcommand{\tablename}{Table}

\def \bbP {\text{$\mathbb{P}$}}
\def \bbE {\text{$\mathbb{E}$}}
\def \bbR {\text{$\mathbb{R}$}}

\DeclareMathOperator{\Cov}{Cov}

\DeclareMathOperator{\Var }{Var}

\DeclareMathOperator*{\argmax}{{arg\,max}}

\DeclareMathOperator{\vecc}{vec}

\newcommand{\lrc}[1]{\left\{#1\right\}}
\newcommand{\lrb}[1]{\left[#1\right]}
\newcommand{\lrp}[1]{\left(#1\right)}

\def \calN {\mathcal N}

\def \calU {\mathcal U}

\def \cvec {\text{\boldmath$c$}}    
    
\def \evec {\text{\boldmath$e$}}

\def \uvec {\text{\boldmath$u$}}    
\def \vvec {\text{\boldmath$v$}}    
\def \wvec {\text{\boldmath$w$}}    
\def \xvec {\text{\boldmath$x$}}    
\def \yvec {\text{\boldmath$y$}}    
\def \zvec {\text{\boldmath$z$}}

\def \mB {\text{\boldmath$B$}}

\def \mG {\text{\boldmath$G$}}

\def \mK {\text{\boldmath$K$}}

\def \mR {\text{\boldmath$R$}}

\def \mU {\text{\boldmath$U$}}
\def \mV {\text{\boldmath$V$}}

\def \mX {\text{\boldmath$X$}}

\def \mZ {\text{\boldmath$Z$}}

\def \ztildevec {\text{\boldmath$\tilde z$}}

\def \mtildeZ {\text{$\widetilde{\mZ}$}}

\def \varepsilonvec   {\text{\boldmath$\varepsilon$}}

\def \etavec          {\text{\boldmath$\eta$}}
\def \thetavec        {\text{\boldmath$\theta$}}

\def \lambdavec       {\text{\boldmath$\lambda$}}
\def \muvec           {\text{\boldmath$\mu$}}
\def \nuvec           {\text{\boldmath$\nu$}}

\def \sigmavec        {\text{\boldmath$\sigma$}}

\def \varphivec       {\text{\boldmath$\varphi$}}
\def \psivec          {\text{\boldmath$\psi$}}

\def \zerovec {\mathbf{0}}
\def \onevec {\mathbf{1}}

\def\theequation{\thesection.\arabic{equation}}  
\def\abstract{\if@twocolumn
\section*{Abstract}
\else \normalsize 
\begin{center}
{\bf Summary\vspace{-.5em}\vspace{0pt}} 
\end{center}
\quotation 
\fi}
\def\endabstract{\if@twocolumn\else\endquotation\fi}

\makeatletter
\newcommand{\myappendix}[1]{
	\setcounter{section}{1}
        \renewcommand{\thesection}{A\arabic{section}}}

\usepackage{color}
\usepackage{colordvi}
\newlength{\breite}
\breite\textwidth
\newcounter{aufg}[section]
  {\refstepcounter{aufg}\noindent\textbf{Exercise \arabic{aufg}:}
   \\*[1ex]\noindent}{\vspace{.5cm}}
   
 \newcounter{notes}[section]
  {\refstepcounter{aufg}\noindent\textbf{}
   \\*[1ex]\noindent}{\vspace{.5cm}}
   
\usepackage{amsthm}  
\newtheorem{assumption}{Assumption}

\newtheorem{theorem}{Theorem}

\newtheorem{definition}{Definition} 

\newtheorem{remark}[]{Remark}
\newtheorem{lemma}[]{Lemma}
\newtheorem*{beisp*}{Example}
\newtheorem{Proof}{Proof}


\newtheoremstyle{break}
  {}
  {}
  {}
  {}
  {\bfseries}
  {.}
  {\newline}
  {}
  
\theoremstyle{break}

\newcommand{\head}[2]%
 {\hrule \vspace{.15cm} {\sfbold Advanced Statistical Inference, Summer Term 2012, Georg-August-University G\"ottingen}\hfill
{\sfbold Sheet #1}\\
{\sfbold Prof. Dr. Thomas Kneib, Nadja Klein}\hfill {\sfbold #2}

\vspace{.2cm}
\hrule

\vspace{1cm}

}

\newcounter{auf}
{\refstepcounter{auf}
\begin{center}
\fcolorbox[gray]{0}{.95}{
\makebox[\breite]{
\textbf{Exercise \arabic{auf}}
}}\\*[1ex]\noindent
\end{center}
}{\vspace{.5cm}}

\newcounter{loes}[section]
{\stepcounter{loes}
\begin{center}
\fcolorbox[gray]{0}{.95}{
\makebox[\breite]{
\textbf{L"osung \arabic{loes}}
}}\\*[1ex]\noindent
\end{center}
}{}

{\begin{center}
\fcolorbox[gray]{0}{.95}{
\makebox[\breite]{
\textbf{Zu Aufgabe #1}
}}\\*[1ex]\noindent
\end{center}\vspace{1cm}
}{\vspace{1cm}}

\newcounter{ka}
{\refstepcounter{ka}
\begin{center}
\framebox[\textwidth]{
\textbf{Aufgabe \arabic{ka}} \hfill #1 Punkte
}\\*[1ex]\noindent
\end{center}
}{\vspace{1cm}}

\newcounter{lka}
{\refstepcounter{lka}
\begin{center}
\framebox[\textwidth]{
\textbf{L\"osung \arabic{lka}} \hfill #1 Punkte
}\\*[1ex]\noindent
\end{center}
}{\vspace{1cm}}

\titlespacing*\section{0pt}{0pt plus 4pt minus 2pt}{0pt plus 2pt minus 2pt}
\titlespacing*\subsection{0pt}{0pt plus 4pt minus 2pt}{0pt plus 2pt minus 2pt}
\titlespacing*\subsubsection{0pt}{0pt plus 4pt minus 2pt}{0pt plus 2pt minus 2pt}
 
\definecolor{myblue}{RGB}{0,73,114}

\renewenvironment{itemize}[1]{\begin{compactitem}#1}{\end{compactitem}}
\renewenvironment{enumerate}[1]{\begin{compactenum}#1}{\end{compactenum}}

\usepackage[skins,breakable]{tcolorbox}
\newtcolorbox{cvbox}[2][]{%
  blanker,
  after skip=8mm,  
  title=#2,
  coltitle=cyan,
  #1 
} 

\makeatletter
\def\@seccntformat#1{%
  \@ifundefined{#1@cntformat}%
    {\csname the#1\endcsname\quad}
    {\csname #1@cntformat\endcsname}%
}
\let\oldappendix\appendix
\renewcommand\appendix{%
  \oldappendix
  \newcommand{\section@cntformat}{\appendixname~\thesection\quad}%
} 
\makeatother

\begin{document}
  \setlength{\abovedisplayskip}{0.15cm} 
  \setlength{\belowdisplayskip}{0.15cm}
  
  \pagestyle{empty}
  \begin{titlepage}

\newcommand{\titlename}{\LARGE\bfseries\color{myblue}  
Vector Vine Copula Models for Multivariate Longitudinal Data}

\title{\titlename}
\author{Michael Stanley Smith and Lin Deng}
\date{\today}
\maketitle
\noindent
{\small Michael Smith is Professor of Management (Econometrics) and Lin Deng is a Postdoctoral Fellow, both at the Melbourne Business School, University of Melbourne, Australia.  Correspondence should be directed to Michael Smith at {\tt mikes70au@gmail.com}. 
\\

\noindent \textbf{Acknowledgments:} Michael Smith's research has been partially supported by the Australian Research Council (ARC) Discovery Project grant DP250101069. This paper uses unit record data from Household, Income and Labour Dynamics in Australia Survey (HILDA) conducted by the Australian Government Department of Social Services (DSS). The findings and views reported in this paper, however, are those of the authors and should not be attributed to the Australian Government, DSS, or any of DSS’ contractors or partners. The authors thank Dr. Weiben Zhang for research assistance in constructing the HILDA dataset.\\}

\newpage
\begin{center}
\mbox{}\vspace{2cm}\\
{\title{\titlename}}\\
\vspace{1cm}
{\Large Abstract \\} 
\end{center}
\vspace{-1pt}
\onehalfspacing
\noindent
Multivariate longitudinal data may exhibit non-Gaussian margins, nonlinear dynamics, and response vectors with composition that varies across waves. To account for these features, we introduce a vector drawable vine (VD-vine) copula that extends conventional drawable vine copulas from scalar to vector-valued nodes. Here, the response vector at each wave forms a multivariate marginal, and serial dependence is captured through a sequence of linking vector copulas. We establish that the VD-vine is itself a vector copula and reduces to a conventional drawable vine for scalar nodes. Recursive forward and backward conditional transports are derived that enable efficient likelihood evaluation and predictive simulation, with parsimonious reductions under finite-order Markov and stationary restrictions. Unconstrained parameterizations for Gaussian and FGM linking vector copulas, flexible multivariate marginals, and Bayesian variational inference provide a practical implementation. Simulations show improved predictive accuracy when the marginals are asymmetric and serial dependence is multivariate, with little loss under a correctly specified Gaussian panel vector autoregression. In an eight-wave Australian panel of 1,093 individuals with varying response vectors, the full VD-vine delivers the best cross-validated distributional forecasts among the models considered, establishing the benefit of capturing asymmetry and nonlinear dependence.
\vspace{15pt}
 
\noindent
{\bf Keywords}: Multivariate Panel Data; Transport Maps; Variational Bayes; Vector Copulas; Vines
\vspace{15pt}

\noindent
{\bf Declaration of AI Usage}: ChatGPT (GPT-5.6 Pro) was used to undertake a pre-submission review of the consistency of the mathematical notation, identification of typographical and minor grammatical errors, and the clarity and succinctness of mathematical expressions and proofs. The authors independently evaluated all suggestions and made manual edits in response to them. The research ideas, mathematical development, and writing (other than as indicated above) are the authors' own. \vspace{15pt}

\noindent
{\bf Data Availability}: The HILDA Survey unit-record data cannot be redistributed under the terms of the data licence. This data is available to authorised researchers through the DSS Longitudinal Studies Dataverse. Instructions on how to construct the dataset used in the application from the unit-record data are provided in~\ref{sm:hildaextra} of the Online Appendix.
\end{titlepage}
  
  \newpage
   
  \pagestyle{plain}
  \setcounter{equation}{0}
  \renewcommand{\theequation}{\arabic{equation}}
  
  
  \section{Introduction}\label{sec:intro}
Multivariate longitudinal data (also called panel data in the social sciences), comprising repeated observations of multiple variables for each individual or other study unit, are common in economics~\citep{HoltzEakinNeweyRosen1988,AttanasioPicciScorcu2000}, finance~\citep{lovezicchino2006}, the social
sciences~\citep{AsparouhovHamakerMuthen2018}, 
medicine and public health~\citep{VerbekeEtAl2014,ZhaoEtAl2021} and elsewhere.
Most existing models for such data have 
temporal and cross-response dependence represented through linear or additive predictors, Gaussian innovations and a fixed number of variables.
But in practice dependencies may be nonlinear, the variable distributions may be strongly non-Gaussian, and the set of variables collected may change across time.
In this paper we propose a new
model for multivariate longitudinal data based on vector copulas~\citep{fan2023vector} that allows for all three features. It generalizes the conventional drawable vine (D-vine) copula model~\citep{aas2009pair} that is
used to capture nonlinear serial dependence in non-Gaussian univariate longitudinal data~\citep{smith2010modeling} to the vector-valued case, so that we call it a ``vector drawable vine'' (VD-vine) copula model. This treats the entire vector of variables observed at each time point as a single multivariate margin and decomposes serial dependence through a sequence of linking vector copulas. While we focus on the VD-vine because it is naturally applicable to multivariate longitudinal data, our approach is general and can extend to other vines, such as R-vines or C-vines~\citep{Czado2019}, to allow for multivariate marginals in a similar fashion.

Conventional vines are high-dimensional copulas constructed 
from bivariate copula components called pair-copulas. Because any copula function can be used for each pair-copula, vines provide a very flexible nonlinear dependence structure; see~\cite{Czado2019} for an overview.
Unfortunately, conventional copulas are not directly
extendable to distributions with given multivariate marginals~\citep{genest1995}.  
However, \cite{fan2023vector} formalize the idea of what they call a ``vector copula'' 
that captures the interdependence between $k$ multivariate marginals, based on transport maps
from each marginal to a product uniform distribution.
In this paper, we generalize pair-copulas to linking vector copulas for two marginals, and use them to construct an extension of a conventional vine to accommodate $k$ multivariate marginals. For the VD-vine, this
allows for localized nonlinear temporal dependence, while accommodating multivariate marginals with composition that can differ across time.
A related model is proposed by~\cite{zhang2026copula}, who use cyclically monotone optimal transport maps to transform each multivariate node margin directly to a standard Gaussian distribution and then estimate an undirected Gaussian graphical model. Their approach also treats an entire vector as a single multivariate marginal, but employs a global Gaussian covariance or precision representation rather than a vector vine factorization.
As far as we are aware, ours is the first paper to use vector copulas
to extend vines to vector-valued nodes.

We establish that the VD-vine constructed in this way is itself a vector copula, 
and show how to compute its density. 
Unlike conventional vines, evaluating the vector vine density requires repeated evaluation of transport maps and their inverses between multivariate conditional distributions. To do so we propose an efficient algorithm that is based on recursive forward and backward transport maps.
Derivation of these recursions is a major contribution of our work, and they
generalize a key result for conventional vines in~\cite{joe1996} to the vector vine case.
Similarly, the proposed algorithm generalizes that for the evaluation of the D-vine density in~\cite{aas2009pair} and~\cite{smith2010modeling}, making likelihood-based inference feasible. The VD-vine also simplifies under finite-order Markov and/or stationary assumptions, reducing computational complexity for long series substantially. 

There are a number of existing approaches that use conventional copulas to capture temporal dependence in multivariate data. One way is to employ conventional vines for a single vector with observed values stacked over both scalar variables and time~\citep{BeareSeo2015,smith2015IJF,BrechmannCzado2015,NaglerKruegerMin2022}. A second way is to model the serial dependence of each response through a separate D-vine and then link the variable-specific conditional distributions using another copula~\citep{ZhaoShiZhang2022,SefidiGanjaliBaghfalaki2022,ShiZhao2024}. A third stream of research employs a latent factor or state dependence structure~\citep{RussoFarcomeni2024,ZitoKowal2025,LiJoeGenest2026}. In contrast, the proposed VD-vine treats the entire vector observed at each time point as a single multivariate margin and decomposes serial dependence through linking vector copulas. This blockwise formulation directly represents nonlinear cross-variable lag dependence, and avoids imposing a scalar-node vine structure over all variable–time coordinates.

Vector copula functions have multivariate marginals that are product uniform distributions, making them more difficult to define than conventional copula functions. We use two different choices for the linking vector copulas. The first is the Gaussian vector copula suggested by~\cite{fan2023vector} and the second is a
novel extension of the Farlie-Gumbel-Morgenstern copula~\citep{JohnsonKotz1975}. In both cases, we propose an unconstrained parameterization that maintains product uniform marginals and is invariant to the ordering of each vector, while being tractable. For the marginals of the vector copula model we consider a 
learnable multivariate generalization of the transformation of~\cite{yeo2000new} that allows for multivariate dependence and skew, although 
other bijective transport maps, such as some normalizing flows, can also be used. 

We estimate the VD-vine using Bayesian variational inference methods, and
their efficacy is demonstrated in a simulation study. It is shown that using the VD-vine model to capture asymmetry in the multivariate marginal distributions, nonlinear serial dependence and cross-wave heterogeneity can all increase the accuracy of density forecasting substantially. 
The effectiveness of the methodology is illustrated in a study of the relationship between work-life conflict and well-being using an Australian panel dataset. This is a panel of $n=1093$ individuals with vector length varying between 4 and 5 over $T=8$ waves. 
We find that accounting for asymmetry in the marginals, along with nonlinear and time-varying  serial dependence, improves accuracy substantially.

The rest of the paper is organized as follows. Section~\ref{sec:vcopmod} specifies 
a vector copula model, Section~\ref{sec:vdvcop} develops the novel VD-vine, and
Section~\ref{sec:panel} outlines its use to model multivariate longitudinal (i.e. panel) data. Section~\ref{sec:vb} shows how to compute likelihood-based inference and prediction efficiently, Sections~\ref{sec:sim} and~\ref{sec:unbalanced} contain the simulation study and empirical application, while Section~\ref{sec:conc} concludes. All proofs and additional material are provided in the Supplementary Material.

  \section{Vector Copula Models}\label{sec:vcopmod}
This section specifies a vector copula model, which is constructed from multivariate marginals and a vector copula function. Suggestions are made for both the marginals and vector copula function that are learnable. In Section~\ref{sec:vdvcop} we develop the VD-vine as a new vector copula that can be used.

\subsection{Definition}
Conventional copula models are based on Sklar's theorem which shows 
that every multivariate distribution can be decomposed into a copula function and its univariate marginals~\cite[pp.46-47]{nelsen06}. A copula model is defined by 
selecting the univariate marginals, along with a copula function, to define a multivariate
distribution. Unfortunately, Sklar's theorem does not apply to multivariate marginals (with the exception of the trivial independence case); see~\cite{genest1995}. However, an analogous decomposition applies to multivariate marginals expressed as transport maps, along with what~\cite{fan2023vector} call a ``vector copula'', as follows.

\begin{definition}[$k$-Vector Copula]\label{def:kvcop}
	Let $\bm{U} = (\bm{U}_1^\top, \bm{U}_2^\top, \dots, \bm{U}_k^\top)^\top$, where each $\bm{U}_j \in [0,1]^{d_j}$ is a subvector of dimension $d_j$, and let $d = \sum_{j=1}^k d_j$ be the dimension of $\bm{U}$. A $k$-vector copula function $C_v : [0,1]^d \to [0,1]$ is a conventional copula function for $\bm{U}$,
	where each subvector is marginally distributed as $\bm{U}_j\sim {\cal U}_j$, with ${\cal U}_j$ a product uniform distribution on $[0,1]^{d_j}$.
\end{definition}

The $k$-vector copula $C_v$ characterizes the dependence between the subvectors $\bm{U}_1, \dots, \bm{U}_k$, while preserving the marginal uniformity of each subvector. 
As with conventional copulas, $c_v(\uvec)= \partial^d C_v(\uvec) / \partial \uvec $ is the vector copula density. 
Theorem~1 of~\cite{fan2023vector} gives the vector copula extension of Sklar's theorem, and we focus
here on its use to construct a vector copula model for $\bm{X} = (\bm{X}_1^\top, \bm{X}_2^\top, \dots, \bm{X}_k^\top)^\top$ with absolutely continuous distribution $\mathbb{P}_X$ on $\mathbb{R}^d$ and 
marginals $\mathbb{P}_{X_j}$ on $\mathbb{R}^{d_j}$ as follows. 

\begin{definition}[$k$-Vector Copula Model]\label{def:vcopmod}
Let $S_j : [0,1]^{d_j} \to \mathbb{R}^{d_j}$ be a measurable bijective transport map that pushes forward ${\cal U}_j$ to marginal $\mathbb{P}_{X_j}$. That is, $(S_j )_{\#} \mathcal{U}_j = \mathbb{P}_{X_j}$,
where ``$\#$'' denotes the push-forward operator. Then 
if $\uvec=(\uvec_1^\top,\ldots,\uvec_k^\top)^\top$,  $S(\uvec)=(S_1(\uvec_1)^\top,\ldots,S_k(\uvec_k))^\top$, $\bm{U}\sim C_v$, 
the joint distribution of $\bm{X}$ is 
given by the push-forward $\mathbb{P}_X=S_{\#}\mathbb{P}_U$, with density 
\begin{equation}\label{eq:vcopdecomp}
f_X(\xvec)=c_v\left(S_1^{-1}(\xvec_1),\ldots,S_k^{-1}(\xvec_k)\right)\prod_{j=1}^k f_j(\xvec_j)\,,
\end{equation}
where $\xvec=(\xvec_1^\top,\ldots,\xvec_k^\top)^\top$ and the marginal density of $\bm{X}_j$ is given by
\[
f_j(\bm{x}_j) = \left| \det \nabla S_j^{-1}(\bm{x}_j) \right|,
\]
with $\nabla S_j^{-1}(\bm{x}_j)$ the Jacobian matrix of the map $S_j^{-1}$ evaluated at $\bm{x}_j$.
\end{definition}

When $d_1=d_2=\ldots= d_k=1$, then 
Definition~\ref{def:vcopmod} is the specification of a $k$-dimensional conventional copula model with univariate marginal distribution functions $F_j(x_j)=S_j^{-1}(x_j)$
for $j=1,\ldots,k$.
Two components are required to use Definition~\ref{def:vcopmod} to construct a vector copula model: the vector copula $C_v$ and the marginal transport maps $S_1,\ldots,S_k$. Suggestions for both that are tractable are given below. 

\begin{remark}\label{rem:cmono}
The canonical construction of~\cite{fan2023vector} takes each $S_j$ to be a cyclically monotone vector quantile map, making $S_j^{-1}$ a vector rank function. However,
cyclical monotonicity is not required for a valid definition of a vector copula model. Any bijective function $S_j$ satisfying $(S_j)_{\#}\mathcal{U}_j=\mathbb{P}_{X_j}$ suffices, although the resulting vector copula is relative to this chosen transport map. We adopt this broader formulation for empirical tractability because imposing cyclical monotonicity on flexible parametric families for $S_j$ is difficult.
\end{remark}

\subsection{Learnable vector copulas}\label{sec:lvcops}
Below are two learnable $k$-vector copulas that 
meet Definition~\ref{def:kvcop}. The first is the Gaussian vector copula proposed by~\cite{fan2023vector}, and the second is a novel extension of the Farlie--Gumbel--Morgenstern (FGM) conventional copula to the vector case.
Efficient parameterization of both for the special case of a linking vector copula, where $k=2$, is discussed later in Section~\ref{sec:vector_dvine_copula}.

\subsubsection{Gaussian Vector Copula (GVC)}\label{sec:gvc}
Define the set of all $(d\times d)$ correlation matrices with identity matrices $I_{d_1},I_{d_2},\ldots,I_{d_k}$ on the block diagonal as ${\cal P}$. 
The GVC is a special case of the popular Gaussian copula~\citep{xue2000}, but where the correlation matrix parameter $\Omega \in {\cal P}$ is patterned.  
The GVC function is given by
\[
C^{\text{GVC}}_v(\bm{u};\Omega) = \Phi_d\left( \underline{\Phi_1}^{-1}(\uvec_1), \dots, \underline{\Phi_1}^{-1}(\uvec_k);\Omega \right)\,.
\]
Here, 
$\Phi_d(\cdot;\Omega)$ is an $ \calN_d(\zerovec,\Omega)$ distribution function, and $\underline{\Phi_1}^{-1}$ is the univariate standard normal quantile function, where the underline denotes that the function is applied element-wise.\footnote{In this paper we use an underline to denote element-wise application of a scalar function to a vector-valued argument.} The corresponding vector copula density is 
\[
c_v^{\text{GVC}}(\uvec;\Omega)=\mbox{det}(\Omega)^{-1/2}\exp\left\{-\frac{1}{2}(\underline{\Phi_1}^{-1}(\uvec))^\top (\Omega^{-1}-I_d)\underline{\Phi_1}^{-1}(\uvec)\right\}\,.
\]
We note that the model of~\cite{zhang2026copula} corresponds to adopting the GVC above
with transports \(S_j=T_j^{-1}\circ\underline{\Phi_1}^{-1}\), where \(T_j\) is the Brenier map from \(\mathbb{P}_{X_j}\) to \(\calN_{d_j}(\zerovec,I)\). 

\subsubsection{Farlie--Gumbel--Morgenstern Vector Copula (FGMVC)}\label{sec:fgmvc}
 The Farlie--Gumbel--Morgenstern (FGM) vector copula is an extension of the conventional FGM copula~\citep{JohnsonKotz1975}. 
 For scalar uniform variables $u_1,\ldots,u_d\in[0,1]$, the $d$-variate conventional FGM copula function is
 \[
 C^{\text{FGM}}(u_1,\ldots,u_d)
 =
 \prod_{l=1}^{d}u_l
 \left[
 1+
 \sum_{\substack{A\subseteq\{1,\ldots,d\}\\ |A|\geq2}}
 \theta_A
 \prod_{l\in A}(1-u_l)
 \right].
 \]
 Each parameter $\theta_A$ corresponds to one interaction subset $A$ with at least two scalar variables. 

 To adapt the FGM copula to the vector copula setting
define the scalar component index set
$\mathcal I=\{(j,i):j=1,\ldots,k,\ i=1,\ldots,d_j\}$, and let
 \[
 \mathcal A_v=\left\{A\subseteq\mathcal I: A \text{ involves {\bf at least two} distinct vector blocks}\right\}.
 \]
 Removing all within-block interaction terms defines the FGM vector copula function
 \[
 C_v^{\text{FGMVC}}(\uvec)
 =
 \prod_{j=1}^{k}\prod_{i=1}^{d_j}u_{ji}
 \left[
 1+
 \sum_{A\in\mathcal A_v}
 \theta_A
 \prod_{(j,i)\in A}(1-u_{ji})
 \right].
 \]
 In the innermost product the scalar index is replaced by the block-component index $(j,i)$, and all interactions that connect at least two vector blocks are retained. To obtain the marginal copula of $\uvec_r$, set $\uvec_j=\onevec_{d_j}$ for all $j\neq r$, in which case $C_v^{\text{FGMVC}}(\uvec)=\prod_{i=1}^{d_r}u_{ri}$ is the independence copula.
 


To reduce the proliferation of parameters as $d$ increases, only pairwise (quadratic) interaction terms are considered. Thus, the vector copula is given by redefining
  \[
  \mathcal A_v=\left\{A\subseteq\mathcal I: \text{contains {\bf exactly two elements}, each {\bf from a different vector block}}\right\}.
  \]
  It is straightforward to show that the vector copula is then
  \begin{equation}\label{eq:fgmvcq}
  C_v^{\text{FGMVC}}(\uvec;\thetavec_M)
  =
  \left( \prod_{j=1}^{k} \prod_{i=1}^{d_j}u_{ji} \right)
  \left[
  1+\sum_{1\leq i<j \leq k} (\onevec_{d_i}-\uvec_i)^\top M_{ij} (\onevec_{d_j}-\uvec_j) \right]\,,
  \end{equation}
  where $M_{ij}\in \mathbb{R}^{d_i\times d_j}$ for $1\leq i<j\leq k$ are parameter matrices with a one-to-one relationship with the unique values of $\theta_A$ for $A\in \mathcal A_v$, and $\thetavec_M=\{M_{ij};1\leq i<j \leq k\}$. The total number of parameters is then $\sum_{1\leq i<j\leq k}d_i d_j$, which is far fewer than in the full FGMVC. A challenge is bounding $\thetavec_M$ to values where~\eqref{eq:fgmvcq} is non-negative, which we address when $k=2$ in Section~\ref{sec:fgm_linking}.

\subsubsection{Discussion of vector copulas}
Enforcing within-block independence makes defining
valid vector copulas more difficult than defining conventional copulas. The implicit copulas of elliptical distributions~\citep{fang2002meta} and skew-elliptical distributions~\citep{joe2019tail} are common choices, with the t~\citep{demarta2005} and skew-t~\citep{deng2024} copulas particularly useful. Apart from the Gaussian case above, extending these to vector copulas is difficult. This is because it requires specification of transport maps that are generalizations of the quantile functions of the multivariate marginals of an elliptical or skew-elliptical distribution, as discussed in~\cite{fan2023vector}. 
Research on the specification of valid vector copulas is limited, and the VD-vine is a new flexible alternative.

\subsection{Learnable multivariate marginals}\label{sec:ltmaps}
Each $S_j^{-1}$ is a bijective transport map, for which we consider a 
flexible form that allows for multivariate dependence and skew in each marginal density $f_j$. It is constructed from a composition of three bijective transformations as follows.
The first transformation maps $\bm{X}_j=(X_{j1},\ldots,X_{jd_j})^\top$ to normality, for which we use the
transformation of~\cite{yeo2000new} (YJ). Let $\etavec_j=(\eta_{j1},\ldots,\eta_{j d_j})^\top$, then the YJ transformation of the standardized random variable $Y_{ji}=(X_{ji}-\mu_{ji})/\sigma_{ji}$ with parameter $\eta_{ji} \in (0,2)$ for each $j,i$ is
\begin{equation}
k_{\eta_{ji}}(y_{ji}) 
=
\begin{cases}
	\displaystyle \frac{(1+y_{ji})^{\eta_{ji}} - 1}{\eta_{ji}}, & y_{ji} \ge 0,\\[1.2ex]
	\displaystyle -\frac{(1-y_{ji})^{2-\eta_{ji}} - 1}{\,2-\eta_{ji}\,}, & y_{ji} < 0.
\end{cases}\label{eq:yj}
\end{equation}
When $\eta_{ji}=1$,~\eqref{eq:yj} is the identity transform, whereas other values of $\eta_{ji}$ account for either positive or negative skew in the distribution of $Y_{ji}$. Setting the diagonal matrix $D_j=\mbox{diag}(\sigma_{j1},\ldots,\sigma_{jd_j})$ and $\muvec_j=(\mu_{j1},\ldots,\mu_{jd_j})^\top$, then the vector-valued transformation applied to $\bm{X}_j$ is
\[
\mK_j=\underline{k_{\etavec_j}}(D_j^{-1}(\bm{X}_j-\muvec_j))=
\left(k_{\eta_{j1}}(Y_{j1}),\ldots,k_{\eta_{jd_j}}(Y_{jd_j})\right)^\top\,.
\]

A second transformation then removes any within block correlation in $\mK_j$ measured by the correlation matrix $\Sigma_j(\varphivec_j)$ which is parameterized in terms of $d_j(d_j-1)/2$ spherical coordinates $\varphivec_j$ that are constrained to ensure $\Sigma_j$ is a correlation matrix; see~\cite{pinheiro1996unconstrained} and~\ref{sm:priors} of the Online Appendix.
 To remove the within-block correlation in the transport, we pre-multiply $\mK_j$ by the spectral matrix operation \((\Sigma_j(\varphivec_j))^{-1/2}  = Q \Lambda^{-1/2} Q^\top\) with spectral decomposition. 

Finally, element-wise application of 
a strictly increasing distribution function $G_\nu$ with parameters $\nuvec$, location zero and unit scale,
maps to uniformity. 
The composition of these three transformations gives 
\begin{equation}
S_j^{-1}(\xvec_j)=\underline{G_\nu}
\left((\Sigma_j(\varphivec_j))^{-1/2}\underline{k_{\etavec_j}}\left(D_j^{-1}(\xvec_j-\muvec_j)\right)
\right)\label{eq:Sj}\,.
\end{equation}
Recall that $\Sigma_j$ is a function of $\varphivec_j$, so that~\eqref{eq:Sj} has 
$(3d_j+d_j(d_j-1)/2+\mbox{dim}(\nuvec))$ parameters $\thetavec_j=\{\nuvec,\muvec_j,\sigmavec_j,\etavec_j,\varphivec_j\}$. 
The bijection $S_j:[0,1]^{d_j} \rightarrow \bbR^{d_j}$ is 
\[
S_j(\uvec_j) = \muvec_j+D_j\underline{k_{\etavec_j}^{-1}}\left(
\Sigma_j(\varphivec_j)^{1/2}\underline{G_\nu}^{-1}(\uvec_j)
\right)\,.
\]

In Sections~\ref{sec:sim} and~\ref{sec:unbalanced} we set $G_\nu=\Phi$ and show how this transformation defines an effective multivariate marginal distribution. We label this distribution as YJN for ``Yeo-Johnson transformed Normal''. It is equivariant to permutation (by which we mean that the distribution is preserved when the data coordinates and the parameters are permuted together) of the elements of $\bm{X}_j$, and when $\etavec_j=\bm{1}$ the YJN is simply a Gaussian. Finally, we stress that other learnable bijective transport maps $S_j^{-1}$ can also be used to define the multivariate marginals.

  \section{Vector Drawable Vine Copula}\label{sec:vdvcop}
\label{sec:vector_dvine_copula}
We now propose a novel $T$-vector copula for $T\geq 2$ that is built from 2-vector copula components.
It generalizes the popular Drawable Vine (D-vine) conventional copula suggested by~\cite{aas2009pair} 
so that we call it a Vector Drawable Vine (VD-vine) copula. Its evaluation
is based on its interpretation as a composition of transport maps.

\subsection{Specification of VD-vine}\label{sec:vcopspecification}
Let $\bm{U}_t\sim {\cal U}_t$ 
denote
a product uniform distribution on $[0,1]^{d_t}$ for $t=1,\ldots,T$. The joint density of $\bm{U}_{1:T}$ \footnote{From Section~\ref{sec:vector_dvine_copula} we adopt standard time series notation for time-indexed variables, where for $a<b$, $\bm{U}_{a:b}$ denotes either the vector $(\bm{U}_a^\top,\ldots,\bm{U}_b^\top)^\top$ or set $\{\bm{U}_a,\ldots,\bm{U}_b\}$ depending on context, whereas if $a>b$, then $\bm{U}_{a:b}\equiv \emptyset$ is the empty set.} is of dimension $d=\sum_{t=1}^T d_t$ and
can be decomposed sequentially as 
\begin{equation}\label{eq:udecomp}
c^\dagger(\uvec_{1:T})=\prod_{t=2}^T f(\uvec_t|\uvec_{1:t-1})\,,
\end{equation}
where $\uvec_t \in [0,1]^{d_t}$ for all $t$, and $f(\uvec_1)=1$. The VD-vine is constructed by assuming, for all $j<t$, a 2-vector copula model
for each conditional distribution $\bm{U}_t,\bm{U}_j|\bm{U}_{j+1:t-1}=\uvec_{j+1:t-1}$ with 2-vector copula function $C_{t,j}$ and density $c_{t,j}$, so that 
\begin{equation}\label{eq:u2vcop}
f(\uvec_t,\uvec_j|\uvec_{j+1:t-1})=c_{t,j} \left( R_{t|j+1}^{-1}(\uvec_t),R_{j|t-1}^{-1}(\uvec_j) \right) f(\uvec_t|\uvec_{j+1:t-1})f(\uvec_j|\uvec_{j+1:t-1})\,.
\end{equation}
This is a $2$-vector copula model defined at~\eqref{eq:vcopdecomp} conditional on the
intervening values $\{\bm{U}_{j+1}=\uvec_{j+1},\ldots,\bm{U}_{t-1}=\uvec_{t-1}\}$. We 
make the simplifying assumption that $C_{t,j}$ is not a function of the intervening values,
which is analogous to the assumption made when constructing a conventional vine copula~\citep{aas2009pair,haff2010simplified}. We call $C_{t,j}$ a ``linking vector copula'', and when $d_j=d_t=1$ it is equivalent to a pair copula used to construct conventional vine copulas.

The measurable transport map $R_{t|j+1}$ pushes forward ${\cal U}_t$ to $\mathbb{P}_{\bm{U}_t|\bm{U}_{j+1|t-1}=\uvec_{j+1|t-1}}$,\footnote{Following convention, $\mathbb{P}_{\bm{X}_1|\bm{X}_2=\xvec_2}$ denotes the probability measure of 
		the conditional distribution $\bm{X}_1|\bm{X}_2=\xvec_2$.} 
and the measurable transport map $R_{j|t-1}$ pushes forward ${\cal U}_j$ to $\mathbb{P}_{\bm{U}_j|\bm{U}_{j+1|t-1}=\uvec_{j+1|t-1}}$. Their inverses  $R_{t|j+1}^{-1}$ and $R_{j|t-1}^{-1}$  specify the corresponding pull-backs. Following 
Remark~\ref{rem:cmono}, $R_{t|j+1}$ and  $R_{j|t-1}$ are not unique here, 
except when $d_t=d_j=1$.
 When $j=t-1$, the conditioning set is defined as empty and the two maps are identity 
transformations with $R_{t|t}(\uvec_t)=\uvec_t$ and $R_{t-1|t-1}(\uvec_{t-1})=\uvec_{t-1}$.
In this special case, because marginally $\bm{U}_t\sim {\cal U}_t$ and  $\bm{U}_j\sim {\cal U}_j$, then the joint density at~\eqref{eq:u2vcop} is exactly equal to $c_{t,j}(\uvec_t,\uvec_j)$.

Re-arranging~\eqref{eq:u2vcop} gives
\begin{equation}\label{eq:ucond}
	f(\uvec_t|\uvec_{j:t-1})=c_{t,j}\left(R_{t|j+1}^{-1}(\uvec_t),R_{j|t-1}^{-1}(\uvec_j)\right)f(\uvec_t|\uvec_{j+1:t-1})\,.
\end{equation}
Setting $j=1$ gives
\[
f(\uvec_t|\uvec_{1:t-1})=c_{t,1}\left(R_{t|2}^{-1}(\uvec_t),R_{1|t-1}^{-1}(\uvec_1)\right)f(\uvec_t|\uvec_{2:t-1})\,.
\] 
Repeated 
application of~\eqref{eq:ucond} with $j=2,3,\ldots,t-1$ leads to the following decomposition for the conditional density
\begin{equation}\label{eq:ucondpdf}
	f(\uvec_t|\uvec_{1:t-1})=\prod_{j=1}^{t-1}c_{t,j}\left(R_{t|j+1}^{-1}(\uvec_t),R_{j|t-1}^{-1}(\uvec_j)\right)\,.
\end{equation}
Substituting~\eqref{eq:ucondpdf} into~\eqref{eq:udecomp} gives the density of the VD-vine copula:
\begin{equation}\label{eq:cvdagger}
c^\dagger(\uvec_{1:T})=\prod_{t=2}^T \prod_{j=1}^{t-1}c_{t,j}\left(R_{t|j+1}^{-1}(\uvec_t),R_{j|t-1}^{-1}(\uvec_j)\right)
=\prod_{t=2}^T \prod_{j=1}^{t-1}c_{t,j}\left(\uvec_{t|j+1},\uvec_{j|t-1}\right)
\end{equation}
which is the product of $(T(T-1)/2)$ linking vector copula densities evaluated at the points 
\[
\uvec_{t|j}:= R^{-1}_{t|j}(\uvec_t)\;\; \mbox{ and }\;\; \uvec_{j|t}:= R^{-1}_{j|t}(\uvec_j)\,,
\]
for $t=2,\ldots,T$ and $j<t$. The VD-vine is also a valid $T$-vector copula, as stated in Lemma~\ref{lem:cvdag} below, with proof in~\ref{sm:grs} of the Online Appendix. 

\begin{lemma}\label{lem:cvdag}
The function $c^\dagger:(0,1)^d  \rightarrow [0,\infty)$ in~\eqref{eq:cvdagger} is the density (defined a.e.) of a valid $T$-vector copula as specified in Definition~\ref{def:kvcop}. 
\end{lemma}

When $d_1=d_2=\cdots=d_T=1$ then~\eqref{eq:cvdagger} is the density of a conventional D-vine. When at least one $d_s\neq d_t$ for some $s,t$, then this decomposition generalizes the D-vine and, as we show later, is applicable to panel data. 

\subsection{Evaluation of VD-vine}
The key to computing $c^\dagger$ at~\eqref{eq:cvdagger} is the evaluation of the arguments ${\cal A}:=\{\uvec_{t|j}, \uvec_{j|t}; t=1,\ldots,T, j\leq t\}$ of the linking vector copulas. Assumption~\ref{asmp:h_fun} below
is helpful in developing an algorithm to do so. 

\begin{assumption}[Conditional Linking Vector Copula Transport]\label{asmp:h_fun}
Let $\bm{V}_t\in [0,1]^{d_t}$ and $\bm{V}_j \in [0,1]^{d_j}$ be jointly distributed with the 2-vector copula
$C_{t,j}$ with density $c_{t,j}$ that is absolutely continuous with $c_{t,j}(\vvec_t,\vvec_j)>0$ for $(\vvec_t,\vvec_j)\in(0,1)^{d_t+d_j}$.
Then assume that the conditional distributions $\bm{V}_t|\bm{V}_j=\vvec_j$ and $\bm{V}_j|\bm{V}_t=\vvec_t$ admit transport maps 
$(H_{t|j}(\cdot|\vvec_j))_{\#} {\mathbb P}_{\bm{V}_t|\bm{V}_j=\vvec_j} ={\cal U}_t$ and
$(H_{j|t}(\cdot|\vvec_t))_{\#} {\mathbb P}_{\bm{V}_j|\bm{V}_t=\vvec_t} ={\cal U}_j$ that push-forward the conditional measure to joint uniform measures.\footnote{We stress that 
	the transport maps $H_{t|j}^{-1}(\cdot|\uvec_j)$ \& $H_{j|t}^{-1}(\cdot|\uvec_t)$ are not equal to $R_{t|j}(\cdot)$ \& $R_{j|t}(\cdot)$ because they transport $\mathcal{U}_t$ \& $\mathcal{U}_j$ to different probability measures. In practice, the maps $R_{t|j}$ and $R_{j|t}$ are not calculated when evaluating the VD-vine copula, and are simply defined to be consistent with the choice of valid forward and backward conditional transports.} 
Also assume that an inverse transport map $H_{t|j}^{-1}(\cdot|\vvec_j)$ can be constructed such that 
\[
H_{t|j}^{-1}\left(H_{t|j}(\vvec_t|\vvec_j)\mid \vvec_j\right)=\vvec_t\,, \mbox{ for }\vvec_t\in[0,1]^{d_t}
\]
for any fixed $\vvec_j$, and an inverse transport map $H_{j|t}^{-1}(\cdot|\vvec_t)$ can be constructed
such that
\[
H_{j|t}^{-1}\left(H_{j|t}(\vvec_j|\vvec_t)\mid \vvec_t\right)=\vvec_j\,, \mbox{ for }\vvec_j\in[0,1]^{d_j}
\]
for any fixed $\vvec_t$.
\end{assumption}

Assumption~\ref{asmp:h_fun} requires the linking vector copula to admit tractable conditional transport maps $H_{t|j}$ and $H_{j|t}$ and their inverse transports, which need not be unique. In general, $H_{t|j}$ and $H_{j|t}$ are defined separately because they are unequal for most choices of vector copulas 
when $d_t>1$ and/or $d_j>1$. 
When $d_j=d_t=1$, so that $c_{t,j}$ is a bivariate conventional pair-copula, then the maps are given by the conditional distribution functions
$H_{t|j}(v_1|v_2)=\frac{\partial}{\partial v_2} C_{t,j}(v_1,v_2)$ and $H_{j|t}(v_2|v_1)=\frac{\partial}{\partial v_1}C_{t,j}(v_1,v_2)$ labeled as `$h$-functions' in~\cite{aas2009pair};  
but when either $d_j>1$ and/or $d_t>1$ they are not.

Theorem~\ref{thm:recursive} below gives recursive relationships between the elements of ${\cal A}$. This generalizes a widely known result due to~\cite{joe1996} for vines to the VD-vine case. It can be used to evaluate
the values of ${\cal A}$ in Algorithm~\ref{thm:recursive} below, which is important because it enables efficient evaluation of the likelihood.

\begin{theorem}[Recursive Evaluation of Arguments]\label{thm:recursive}
	Let Assumption~\ref{asmp:h_fun} be satisfied, the VD-vine density \( c^\dagger(\boldsymbol{u}_{1:T}) \) be specified as in~\eqref{eq:cvdagger}, and that $\uvec_{t|t}\equiv \uvec_t$ for all $t=1,\ldots,T$. 
	Then, for all \( t = 2, \dots, T \) and all \( j < t \), the following recursions hold:
	\begin{enumerate}[(i)]
		\item \text{Forward Conditional Transport Map: }
		$\boldsymbol{u}_{t|j} =  H_{t|j}(\boldsymbol{u}_{t|j+1} \mid \boldsymbol{u}_{j|t-1})$,
		\item \text{Backward Conditional Transport Map: }
		$\boldsymbol{u}_{j|t} = H_{j|t}(\boldsymbol{u}_{j|t-1} \mid \boldsymbol{u}_{t|j+1})$.
	\end{enumerate}
\end{theorem}
 
The algorithm below employs these recursive relationships to evaluate ${\cal A}$. 
The density $c^\dagger$ can then be computed as the product at~\eqref{eq:cvdagger} using
the elements ${\cal A}$ as the linking vector copulas arguments.
 
\begin{algorithm}[H]
	\caption{Recursive Evaluation of ${\cal A}$}\label{alg:recursive}
	\begin{algorithmic}
		\State Input: $\uvec_1,\ldots,\uvec_T$
		\State \text{Step 0: Initialization}
		\For{$t = 1$ \text{to} $T$}
		\State Set $\uvec_{t|t} := \uvec_t$
		\EndFor
		
		
		\State Step 1: Nested Recursion
		\For{$r = 1,\ldots,T-1$}
		\For{$j=1,\ldots,T-r$}
		\State Set $t=r+j$ and compute:
		\[
		\uvec_{t|j} = H_{t|j}(\uvec_{t|j+1} \mid \uvec_{j|t-1}), \quad \mbox{ and }\quad
		\uvec_{j|t} = H_{j|t}(\uvec_{j|t-1} \mid \uvec_{t|j+1})
		\]
		\EndFor
		\EndFor
		
		
		\State Output: ${\cal A}=\{\uvec_{t|j},\uvec_{j|t}; t=1,\ldots,T,\,\, j\leq t\}$
		
	\end{algorithmic}
\end{algorithm}

The ordering (i.e. formation of the loops) for the recursive computations in Algorithm~\ref{alg:recursive} is important for 
three reasons. First, it ensures that the arguments of $H_{t|j}$ and $H_{j|t}$ are available to calculate each forward and backward recursion. Second, it arranges the dependencies such that the inner loop in $j$ can be undertaken in parallel. Third, as discussed further below, it allows for the computations to be easily simplified for Markov processes.   

\subsection{Linking vector copulas}\label{sec:linkvcs}
We use the two vector copulas in Section~\ref{sec:lvcops} with $k=2$ as linking vector copulas in the VD-vine. Key to employing each is deriving a tractable
vector copula parameterization and suitable prior, as discussed below. Appendix~\ref{app:2vcop} specifies the
corresponding conditional transport maps and their inverses required to implement Algorithm~\ref{alg:recursive}.

\subsubsection{GVC Parameterization}

When $k=2$, the vector copula
blocks are of size $d_1$ and $d_2$, and the parameter matrix of the GVC is 
\begin{equation*}
\Omega=
\begin{bmatrix}
I_{d_1} & M^\top\\
M & I_{d_2}\\
\end{bmatrix}
\end{equation*}
The condition that $\Omega\in \mathcal{P}$ is a correlation matrix
is equivalent to the operator norm condition that
$\|M\|_{\mathrm{op}}<1$~\citep[Sec.~7.7]{horn2012matrix}. A convenient one-to-one re-parameterization of $M$ to an unconstrained matrix $B\in \mathbb{R}^{d_2\times d_1}$ that satisfies this condition 
is suggested in the lemma below.

\begin{lemma}\label{lem:free_param}
	Let $\delta>0$ and 
	${\cal M}_\delta=\{M\in\mathbb{R}^{d_2\times d_1}:\|M\|_{\mathrm{op}}<\delta\}$. Then 
	the map \(\Psi_\delta:{\cal M}_\delta\rightarrow\mathbb{R}^{d_2\times d_1}\) defined by
	\[
	B=\Psi_\delta(M)=\frac{1}{\delta}M(I_{d_1}-\delta^{-2}M^\top M)^{-1}
	\]
	is a bijection, and its inverse is
	\[
	\Psi_\delta^{-1}(B)
	=
	2\delta B\bigl(I_{d_1}+(I_{d_1}+4B^\top B)^{1/2}\bigr)^{-1},
	\]
	where \((I_{d_1}+4B^\top B)^{1/2}\) denotes the principal matrix square root.
	In particular, the inverse transformation always satisfies the operator norm condition
	$\|\Psi_\delta^{-1}(B)\|_{\mathrm{op}}<\delta$.
\end{lemma}


Setting $\delta=1$ in Lemma~\ref{lem:free_param} provides an unconstrained parameterization $B=M(I_{d_1}-M^\top M)^{-1}$ for GVC. 

\subsubsection{FGMVC Parameterization}\label{sec:fgm_linking}

We also consider the case where FGM vector copulas are used as linking copulas in the VD-vine construction. 
Let $\vvec_j=(v_{j1},\ldots,v_{jd_j})^\top\in[0,1]^{d_j}$ for $j=1,2$, then with $k=2$ 
\[
C^{\text{FGMVC}}_{v}(\vvec_1,\vvec_2;M)
=
\left(\prod_{j=1}^2 \prod_{i=1}^{d_j} v_{ji} \right)
\lrb{
1+(\onevec_{d_2}-\vvec_2)^\top M(\onevec_{d_1}-\vvec_1)
},
\]
with $M\in \mathbb{R}^{d_2\times d_1}$ (so that $\thetavec_M=M_{12}=M^\top$ in~\eqref{eq:fgmvcq}) and the corresponding copula density 
\begin{equation}\label{eq:fgmvc_dens}
c^{\text{FGMVC}}_v(\vvec_1,\vvec_2;M)
=
1+(\onevec_{d_2}-2\vvec_2)^\top M(\onevec_{d_1}-2\vvec_1).
\end{equation}
The admissible parameter space is determined by the non-negativity condition $1+\zvec_2^\top M \zvec_1\geq 0$ for all $\zvec_j\in[-1,1]^{d_j}$. 
A simple sufficient condition is $\|M\|_{\mathrm{op}}< 1/\sqrt{d_1 d_2}$, since $|\zvec_2^\top M\zvec_1|\leq \|\zvec_1\|_2\|M \|_{\mathrm{op}}\|\zvec_2\|_2\leq \sqrt{d_1 d_2}\|M\|_{\mathrm{op}}$. As with the GVC, Lemma~\ref{lem:free_param} can be used to define an unconstrained parameter matrix $B$ that satisfies this operator condition by setting $\delta=1/\sqrt{d_1 d_2}$. Therefore, the parameterization is given by $B=\sqrt{d_1 d_2}M(I_{d_1}-d_1d_2M^\top M)^{-1}\in \mathbb{R}^{d_2\times d_1}$, with inverse $M=\Psi_{(d_1d_2)^{-1/2}}^{-1}(B)$. 

  \section{Multivariate Dynamic Panel Model}\label{sec:panel}
D-vine copulas are popular for capturing serial dependence in univariate time
series~\citep{smith2010modeling} and the proposed VD-vine extends this to capturing serial dependence in multivariate time series. 
It can be used to define a multivariate dynamic panel model that generalizes existing PVAR models~(see~\citealp{Canova2013}) to allow for nonlinear dependence as outlined below. 

\subsection{Multivariate time series}\label{sec:multTS}
Consider a vector stochastic process $\{\bm{X}_t\}$ with $\bm{X}_t\in \mathbb{R}^{d_t}$, which
may differ in composition of elements, so that it is possible that $d_t\neq d_j$ for 
one or more pairs $t\neq j$. 
The VD-vine copula is used to model the multivariate temporal dependence in this process as follows. Following Definition~\ref{def:vcopmod}, for each $t$ consider a transport map $S_t:[0,1]^{d_t}\rightarrow \mathbb{R}^{d_t}$ that pushes
a uniform product measure ${\cal U}_t$ onto that of the marginal of $\bm{X}_t$, so that $(S_t)_{\#}{\cal U}_t=\mathbb{P}_{\bm{X}_t}$. In our work we use a
learnable transformation of the form given in Section~\ref{sec:ltmaps} for each $S_1,\ldots,S_T$. 
Then the joint distribution of $\bm{X}_{1:T}$ is given by the $T$-block VD-vine copula model with density as in Definition~\ref{def:vcopmod}:
\begin{equation}\label{eq:fX1toT}
f_{1:T}(\xvec_{1:T})=c^{\dagger}\left(S_1^{-1}(\xvec_1),\ldots,S_T^{-1}(\xvec_T)\right)\prod_{t=1}^T f_t(\xvec_t)\,,
\end{equation}
where $c^\dagger$ is the VD-vine copula density.
Assuming $S_t$ is a measurable differentiable bijective function, 
then the marginal density of $\bm{X}_t$ is 
$f_t(\bm{x}_t) = \left| \det \nabla S_t^{-1}(\bm{x}_t) \right|$.
Evaluation of the term $c^\dagger$ above  involves computing in order: (i)~$\uvec_t=S_t^{-1}(\xvec_t)$ for $t=1,\ldots,T$, (ii)~the arguments ${\cal A}$ using Algorithm~\ref{alg:recursive}, and finally (iii)~the product
at~\eqref{eq:cvdagger}. 

It also follows from~\eqref{eq:ucondpdf} that the conditional distribution of  $\bm{X}_t|\bm{X}_{1:t-1}=\xvec_{1:t-1}$ has density 
\begin{equation}\label{eq:TScondpdf}
	f(\xvec_t|\xvec_{1:t-1})=\prod_{j=1}^{t-1}c_{t,j}\left(\uvec_{t|j+1},\uvec_{j|t-1}\right)f_t(\xvec_t)\,.
\end{equation}
Note that for all $j\leq t-1$ (i.e. the terms in the product), the first argument of $c_{t,j}$ is $\uvec_{t|j+1}=R_{t|j+1}^{-1}(\uvec_t)=R_{t|j+1}^{-1}(S_t^{-1}(\xvec_t))$ is a function of 
$\xvec_t$. However, the second argument $\uvec_{j|t-1}=R_{j|t-1}^{-1}(\uvec_j)=R_{j|t-1}^{-1}(S_j^{-1}(\xvec_j))$ with $j<t$ is not dependent on $\xvec_t$. A transport map that pushes forward from the distribution with density given at~\eqref{eq:TScondpdf} to the uniform ${\cal U}_t$, can be obtained by repeated application of the forward transport
map in Theorem~\ref{thm:recursive} to obtain the composition map as outlined below.

\begin{lemma}[Composition Forward Transport Map]\label{lem:compforwardmap}
For $j=1,\ldots,t-1$, let $H_{t|j}(\cdot)$ be shorthand for the forward transport map $H_{t|j}(\cdot|\uvec_{j|t-1})$ defined as in Assumption~\ref{asmp:h_fun}, then the composition forward transport maps
\begin{eqnarray}
\left(H_{t|1}\circ H_{t|2}\circ \cdots \circ H_{t|t-1}\right)_{\#}  \mathbb{P}_{\bm{U}_t|\bm{U}_{1:t-1}=\uvec_{1:t-1}} &=&{\cal U}_t\,,\label{eq:ctransportu}\\ 
\left(H_{t|1}\circ H_{t|2}\circ \cdots \circ H_{t|t-1}\circ S_t^{-1}\right)_{\#}  \mathbb{P}_{\bm{X}_t|\bm{X}_{1:t-1}=\xvec_{1:t-1}} &=&{\cal U}_t \,.\label{eq:conditionaltransport}
\end{eqnarray}
Here, the arguments $\uvec_{1|t-1},\ldots,\uvec_{t-1|t-1}$ are dependent on $\xvec_{1:t-1}$ and can be obtained by 
applying Algorithm~\ref{alg:recursive}, but with the outer loop running $r=1,\ldots,t-1$, rather than to $T$.
\end{lemma}

Thus, to obtain a draw from the conditional at~\eqref{eq:TScondpdf}, first draw $\uvec_t \sim {\cal U}_t$, then apply the inverse of~\eqref{eq:conditionaltransport}:
\begin{equation}
	\xvec_t= \left\{ \begin{array}{lr} S_t(\uvec_t) &\mbox{ if }t=1\\
	S_t \circ H_{t|t-1}^{-1} \circ H_{t|t-2}^{-1} \circ \cdots \circ H_{t|1}^{-1}(\uvec_t|\uvec_{1|t-1}) &\mbox{ if }t>1 \end{array}\right.\,, \label{eq:simx}
\end{equation}
where each inverse transport map $H_{t|j}^{-1}(\cdot|\uvec_{j|t-1})$ conditions on $\uvec_{j|t-1}$; see Appendix~\ref{app:2vcop} for these inverses for the 2-GVC and 2-FGMVC
linking copulas.

\subsection{Markov and stationary processes}
\subsubsection{Markov process}\label{sec:markovdef}
If the stochastic process is Markov of order $p\geq 1$, then for all $j<t-p$ the linking vector densities $c_{t,j}(\vvec_1,\vvec_2)=1$ (i.e. it is an independence vector copula). 
Thus, if $t_0=\max(1,t-p)$, the conditional
density at~\eqref{eq:TScondpdf} simplifies to
\[
f(\xvec_t|\xvec_{t-1},\ldots,\xvec_1)=\prod_{j=t_0}^{t-1}c_{t,j}\left(\uvec_{t|j+1},\uvec_{j|t-1}\right)f_t(\xvec_t)\,.
\]
When drawing from this conditional note that $H_{t|j}(\uvec_t|\uvec_j)=\uvec_t$ for all $j<t-p$ when evaluating~\eqref{eq:simx}.
Similarly, the lower limit of the inner product in~\eqref{eq:cvdagger} is $j=t_0$, which greatly reduces the number of linking copulas that need evaluating to compute $c^\dagger$ when $p<<T$. 

\subsubsection{Stationary process} 
When the vector composition is constant over time, the VD-vine can be constrained to define multivariate strictly stationary processes on both the copula and data scales, as follows.
\begin{lemma}\label{lem:station}
Let $\{{\bm U}_t\}$ and $\{{\bm X}_t\}$ follow the stochastic processes with densities at~\eqref{eq:cvdagger} and~\eqref{eq:fX1toT} for observations at times $1,2,\ldots,T$, and fixed composition so that $d_t=d_{t'}$ for all $t,t'$. Also, let Assumption~\ref{asmp:h_fun} hold and ${\cal L}_{t,j}:=(C_{t,j},H_{t|j},H_{j|t},H_{t|j}^{-1},H_{j|t}^{-1})$ define the orientated linking copula specification. Then
\begin{itemize}
	\item[(i)] 
 the process $\{{\bm U}_t\}$ is strictly stationary if for all $\ell=1,\ldots,T-1$, the linking specifications ${\cal L}_{t,t-\ell}={\cal L}_{t',t'-\ell}$ for all $\ell+1\leq t,t'\leq T$, 
 \item[(ii)] the process $\{{\bm X}_t\}$ is also strictly stationary if both~(i) holds and the transport maps $S_t=S_{t'}$ for all $t,t'$.  
 \end{itemize}
\end{lemma}
 Lemma~\ref{lem:station} extends the stationary D-vine of~\cite{smith2015IJF} and the translation-invariant stationary vine framework of~\cite{NaglerKruegerMin2022}. For vectors with composition that changes over time, multivariate stationarity is difficult to define and enforce without additional model structure.

\subsubsection{Connection to partial correlation matrices} Let all linking copulas in a stationary VD-vine be 2-GVCs, and consider the latent Gaussian stochastic process $\{\bm{Z}_t\}$ where $\bm{Z}_t=\underline{\Phi}^{-1}(\bm{U}_t)$. In this case, the parameter matrix $M_{t,t-\ell}$ of the lag-$\ell$ linking copula is the lag-$\ell$ partial correlation matrix of $\{\bm{Z}_t\}$ as discussed by~\cite{AnsleyKohn1986} and~\cite{Heaps2023}. 
Finally, for the case where $d_1=d_2=\ldots=d_T=1$, the matrix $M_{t,t-\ell}$ is the scalar partial
correlation for a time series as noted by~\cite{smith2010modeling}.

  \section{Bayesian Inference}\label{sec:vb}
Exact posterior inference can be computed for conventional D-vine copula models using 
Markov chain Monte Carlo (MCMC) methods as discussed in~\cite{MinCzado2010,MinCzado2011} and~\cite{smith2010modeling}. However, for large vine copula models these can be slow, so as a scalable alternative, variational inference (VI) methods have been used for vine copula models by~\cite{LoaizaMayaSmith2019} and \cite{KejzlarMaiti2022}.
Exact posterior inference can also be even slower for the VD-vine copula model proposed here because
the parameters $(\thetavec,\psivec)$ can be high dimensional
(e.g. $(\#(\thetavec),\#(\psivec))=(179,598)$ in the HILDA panel application). 
Therefore, we also employ VI to compute an approximate posterior for the VD-vine copula model as now outlined. 

\subsection{Variational Inference}
VI methods are popular for computing Bayesian inference for large stochastic
models for which it is difficult to compute the exact posterior distribution~\citep{blei2017variational,Zhang2019AdvancesVI}. Let $\mathcal{D}$ denote the observed
data, $p(\mathcal{D}|\thetavec,\psivec)$ the likelihood, and $p(\thetavec,\psivec)$ the prior. Then VI approximates the posterior distribution $p(\thetavec,\psivec|\mathcal{D})\propto p(\mathcal{D}|\thetavec,\psivec)p(\thetavec,\psivec):=h(\thetavec,\psivec)$ using a 
 density $q_\lambda(\thetavec,\psivec)\in Q$ in a variational family $Q$ indexed by 
 $\lambdavec$ that are called the variational parameters. The idea is to identify a value of $\lambdavec$ so that $q_\lambda$ is close to $p(\thetavec,\psivec|\mathcal{D})$, with distance usually measured by the Kullback-Leibler divergence (KLD). It is straightforward to show that minimizing the KLD is equivalent to maximizing the Evidence Lower Bound (ELBO) function as follows:
 \begin{equation}
 	\lambdavec^\star = \argmax_{\lambdavec} \mathbb{E}_{q_\lambda}\left[\log h(\thetavec,\psivec)-\log q_\lambda(\thetavec,\psivec)\right]\,.\label{eq:elbo}
 \end{equation}
 Key to effective VI is the choice of a family $Q$ with density $q_\lambda$ which is sufficiently expressive, and amenable to fast implementation of optimization methods
 to solve~\eqref{eq:elbo}. 
 Here, we follow~\cite{ong2018gaussian} and employ a Gaussian approximation with a factor decomposition, which satisfies both requirements. This is a popular choice and these authors give details on the efficient implementation of VI, which we use here and do not describe. Our empirical work uses five factors in the approximation, which these authors found to be sufficient for accurate approximations to a variety of posteriors in practice.
 Alternative variational approximations, 
 such as normalizing flows or even other copula models, can also be used here. Code to implement these methods is widely available.

\subsection{Longitudinal Likelihood}
Multivariate panel data consist of longitudinal observations on 
a multivariate time series over units, such as individuals or households. Define $\xvec_{tj}^i$ as the observation of variable $j\in \{1,2,\ldots,d_t\}$ at time point $t\in\{1,2,\ldots,T\}$ for unit $i\in\{1,\ldots,n\}$, and let
$\xvec_{t}^i:=(\xvec_{t1}^i,\ldots,\xvec_{td_t}^i)^\top$, $\xvec_{s:t}^i=\{\xvec_{s}^i,\xvec_{s+1}^i,\ldots,\xvec_{t}^i\}$ for $s<t$. Denote the parameters of each
learnable marginal transformation $S_j$ as $\thetavec_j$, and set $\thetavec=(\thetavec_1^\top,\ldots,\thetavec_T)^\top$. Denote the parameters of each linking vector copula $c_{t,j}$ as 
$\psivec_{t,j}$, and let $\psivec$ be the vector comprising the parameters $\{\psivec_{t,j};t=2,\ldots,T,\; j<t\}$ of all linking vector copulas. 
Following the same convention of denoting with superscript ``$i$'' quantities that correspond to a longitudinal observation, then
the likelihood is
$p(\mathcal{D}|\thetavec,\psivec)=\prod_{i=1}^n {\mathcal L}^i(\xvec^i_{1:T};\thetavec,\psivec)$ with
\begin{equation}
	 {\mathcal L}^i(\xvec^i_{1:T};\thetavec,\psivec)= c^\dagger(\uvec_{1:T}^i;\psivec)\prod_{t=1}^T f(\xvec_{t}^i;\thetavec_t)=
	\prod_{t=2}^T\prod_{j<t}c_{t,j} \lrp{\uvec^i_{t|j+1},\uvec^i_{j|t-1};\psivec_{t,j}}\prod_{t=1}^T f(\xvec_{t}^i;\thetavec_t)\,.\label{eq:longlike}
\end{equation}
Here, each set of arguments ${\cal A}^i:=\{\uvec_{t|j}^i,\uvec_{j|t}^i;t=1,\ldots,T,j<t\}$ is computed using Algorithm~\ref{alg:recursive} applied to $\uvec^i_{1:T}=(S_1^{-1}(\xvec_{1}^i;\thetavec_1)^\top,\ldots,S_T^{-1}(\xvec_{T}^i;\thetavec_T)^\top)^\top$. 
To implement gradient-based estimation methods we compute the gradient of the log posterior $\log p(\thetavec,\psivec|\mathcal{D})= \log h(\thetavec,\psivec) + {\cal C}=\sum_i\log{\mathcal L}^i(\xvec^i_{1:T};\thetavec,\psivec)+\log p(\thetavec,\psivec)+{\cal C}$ up to constant ${\cal C}$ using 
automatic differentiation in PyTorch.

\subsection{Priors}
\label{sec:priors}

\subsubsection{Prior for each linking vector copula}
For both the 2-GVC and 2-FGMVC linking copulas, dependence is parameterized by a matrix $M$, subject to the restriction $\lVert M\rVert_{\mathrm{op}}<\delta$ with $\delta=1$ for the GVC and $\delta=(d_1d_2)^{-1/2}$ for the FGMVC. Lemma~\ref{lem:free_param} defines a unconstrained matrix \(B\) that satisfies this constraint, for which we use the Gaussian prior 
$\vecc(B) \sim \calN(0,\tau^2 I_{d_1d_2})$ with $\tau^2=25$ in our empirical work. It is straightforward in our framework to also treat $\tau^2$ as a hyperprior, or 
consider some other conditionally Gaussian regularization priors such as the horseshoe~\citep{carvalho2010horseshoe}.

For the 2-GVC, the implied prior on $M$ is 
\begin{equation*}
	p_M(\vecc M) =\left\{  
	\begin{array}{ll}
		\phi\lrp{\vecc \lrp{M(I - M^\top M)^{-1} },\zerovec,\tau^2 I_{d_1d_2} }\,|J(M)| &\mbox{ if } \|M\|_{\mathrm{op}}<1 \\
		0 &\mbox{ otherwise}\,,
	\end{array}
	\right.
\end{equation*}
where \(\phi(\,\cdot\,;\boldsymbol{\mu},\Sigma)\) denote the Gaussian density function with mean $\muvec$ and covariance matrix $\Sigma$, and $|J(M)|$ is the determinant of the Jacobian term. 
Let \(\gamma_1,\dots,\gamma_r\) denote the singular values of \(M\), where
$r=\min(d_1,d_2)$, then when \(\|M\|_{\mathrm{op}}<1\) this determinant can be computed analytically as
\[|J(M)|= \left| \det\!\left(
\frac{\partial \,\vecc B}{\partial \,\vecc M^\top} \right)
\right|
=
\prod_{i=1}^r
\frac{1+\gamma_i^2}{(1-\gamma_i^2)^{d_1+d_2}}
\;\prod_{1\le i<j\le r}(1-\gamma_i^2\gamma_j^2).
\]
The behavior of this prior on $M$ differs markedly across its two spectral regimes. 
When all singular values of $M$ are close to zero the prior density is smoothly equivalent to a Gaussian in $M$. In contrast, when $\|M\|_{\mathrm{op}} $ approaches boundary \(\delta\), the prior strongly suppresses mass near the boundary of the admissible space. Under mild regularity of the likelihood, this behavior discourages posterior concentration near this boundary and contributes to numerical stability in estimation. An analysis of the implied prior on $M$ for the 2-FGMVC is similar. In either case, evaluation 
of this prior or the posterior of $M$ is unnecessary for inference, although evaluation of $M=\Psi^{-1}_\delta(B)$ is required.

\subsubsection{Prior for each marginal}
Because we set $G_\nu = \Phi$, the parameters of each YJN marginal are \(\thetavec_j=(\muvec_j^\top,\sigmavec_j^\top,\etavec_j^\top,\varphivec_j^\top)^\top\). 
To facilitate VI these are transformed to unconstrained parameterizations as follows. 
We employ the logarithm of each element of $\sigmavec_j$, transform each element of $\etavec_j$ from $(\epsilon_\eta,2-\epsilon_\eta)$ using a scaled logit transform, and each hyper-spherical
angle from $(\epsilon_\varphi,\pi-\epsilon_\varphi)$ using a scaled logit transform; $\muvec_j$ is not transformed. 
The small numbers $\epsilon_\eta=10^{-3}$ and $\epsilon_\varphi=0.03$ are set to ensure 
numerical stability.
Independent proper priors for each scalar element of $\thetavec_j$ are used. The priors on the transformed parameters in our empirical analysis are all zero mean normals, with a standard deviation of 5 for the elements of $\muvec_j$, 0.5 for the elements of $\sigmavec_j$, and unity for the elements of $\etavec_j$ and $\varphivec_j$.

\subsection{Predictive Simulation}\label{sec:predsim}
Panel wave predictions at origin $t$ for $h\geq 1$ steps ahead with $t+h\leq T$ are constructed via Monte Carlo simulation from the joint predictive distribution on the copula scale $\bm{U}_{t+1:t+h}|\bm{U}_{1:t}=\uvec_{1:t}$. The marginal predictive distribution of $\bm{U}_{t+h}|\bm{U}_{1:t}=\uvec_{1:t}$ is obtained by
retaining the draws of $\bm{U}_{t+h}$. 
 From~\eqref{eq:udecomp}, simulation from the joint corresponds to generating sequentially from each 
conditional  $\bm{U}_{s}|\bm{U}_{1:s-1}$ for $s=t+1,\ldots,t+h$.
At each step $s$, condition on a realized path
\(\uvec_{1:s-1}^{\star}\), where \(\uvec_j^{\star}=\uvec_j\) for
\(j\leq t\) and \(\uvec_j^{\star}\) is simulated for \(j>t\). By
Lemma~\ref{lem:compforwardmap}, with
\(H_{s|j}(\cdot)\) denoting
\(H_{s|j}(\cdot\mid\uvec_{j|s-1}^{\star})\),
\[
\left(
H_{s|1}\circ H_{s|2}\circ\cdots\circ H_{s|s-1}
\right)_{\#}
\mathbb{P}_{\bm{U}_s\mid \bm{U}_{1:s-1}=\uvec_{1:s-1}^{\star}}
=
{\cal U}_s .
\]
Hence the inverse composition maps an independent
\(\wvec_s\sim{\cal U}_s\) into a draw from
\(\mathbb{P}_{\bm{U}_s\mid \bm{U}_{1:s-1}=\uvec_{1:s-1}^{\star}}\).

\begin{algorithm}[H]
	\caption{Joint predictive simulation from
		\(\mathbb{P}_{\bm{U}_{t+1:t+h}\mid \bm{U}_{1:t}=\uvec_{1:t}}\)}
	\label{alg:predsimU}
	\begin{algorithmic}
		\State Input: observed history \(\uvec_{1:t}\) and horizon \(h\geq 1\) for $t+h\leq T$.
		
		\State Set \(\uvec_j^{\star}:=\uvec_j\), \(j=1,\ldots,t\).
		\State Apply Algorithm~\ref{alg:recursive} to \(\uvec_{1:t}^{\star}\) to obtain
		\(\uvec_{j|t}^{\star}\), \(j=1,\ldots,t\), with
		\(\uvec_{t|t}^{\star}\equiv \uvec_t^{\star}\).
		
		\For{\(s=t+1,\ldots,t+h\)}
		\State Draw \(\wvec_s\sim{\cal U}_s\) and set
		$\uvec_{s|1}^{\star}:=\wvec_s$.
		\State Invert the composition transport:
		\For{\(j=1,\ldots,s-1\)}
		\[
		\uvec_{s|j+1}^{\star}
		:=
		H_{s|j}^{-1}
		\left(
		\uvec_{s|j}^{\star}\mid \uvec_{j|s-1}^{\star}
		\right).
		\]
		\EndFor
		
		\State Set \(\uvec_s^{\star}:=\uvec_{s|s}^{\star}\).
		
		\State Update the conditioning arguments required for subsequent future blocks:
		\For{\(j=1,\ldots,s-1\)}
		\[
		\uvec_{j|s}^{\star}
		:=
		H_{j|s}
		\left(
		\uvec_{j|s-1}^{\star}\mid \uvec_{s|j+1}^{\star}
		\right).
		\]
		\EndFor
		\EndFor
		
		\State Output: the future path	$\uvec_{t+1:t+h}^{\star}
		=
		(\uvec_{t+1}^{\star},\ldots,\uvec_{t+h}^{\star})$.
	\end{algorithmic}
\end{algorithm}

Algorithm~\ref{alg:predsimU} yields a draw from the joint
predictive distribution
$\bm{U}_{t+1:t+h}\mid \bm{U}_{1:t}=\uvec_{1:t}$. 
The first inner loop applies the inverse of the composition map in
Lemma~\ref{lem:compforwardmap} to produce a draw from
$\mathbb{P}_{\bm{U}_s\mid \bm{U}_{1:s-1}=\uvec_{1:s-1}^{\star}}$, while the second inner loop is the backward recursion in
Theorem~\ref{thm:recursive} that updates
\(\uvec_{j|s}^{\star}\) as needed for simulating later blocks. 
If only the \(h\)-step-ahead predictive distribution is required, retain
\(\uvec_{t+h}^{\star}\) and discard the intermediate simulated blocks.
If a predictive draw is
required on the original data scale, note that 
$(\bm{X}_{t+1:t+h}|\bm{X}_{1:t}=\xvec_{1:t}) \overset{d}{=}
(\bm{X}_{t+1:t+h}|\bm{U}_{1:t}=\uvec_{1:t})$, so that 
$\xvec_s^{\star}=S_s(\uvec_s^{\star})$ for $s=t+1,\ldots,t+h$.

For a Markov process of order \(p\), the same algorithm is used after omitting
the independence links \(j<s-p\), as discussed in
Section~\ref{sec:vector_dvine_copula}. Equivalently, at step $s$, set
$s_0=\max(1,s-p)$, initialize
$\uvec_{s|s_0}^{\star}:=\wvec_s$ 
and run the inverse and update recursions only over
\(j=s_0,\ldots,s-1\).

  \section{Simulation Study}\label{sec:sim}
We undertake a simulation study to compare the VD-vine copula model to a Gaussian PVAR. It shows that
capturing marginal asymmetry and nonlinear serial dependence using the VD-vine copula model
can improve distributional fit greatly. Yet, when data is generated from a PVAR, the VD-vine copula model is almost as accurate.

\subsection{Data Generating Processes}
We consider two data generating processes with $T=15$ waves and $d_1=\ldots=d_T=5$. The first (DGP1) is a VD-vine copula model, where the vector vine is stationary with Markov $p=2$. The linking
vector copulas are 2-GVCs with $(5\times 5)$ parameter matrices; see~\ref{sm:sim} of the Online Appendix. The marginals are YJN distributions with location zero, unit scale, but skew parameters that vary over $t$, with  
\[
\etavec_t =
\begin{cases}
	\lrp{1,1,1,1,1}^{\top}, & t=1,\\
	\lrp{0.60,1.40,0.70,1.30,1}^{\top}, & t\ \text{even},\\
	\lrp{1.35,0.65,1.30,0.70,1}^{\top}, & t\ \text{odd}, t\geq 3.
\end{cases}
\]
Therefore, while $\{\bm{U}_t\}$ is a stationary process,
$\{\bm{X}_t\}$ is not because the marginals change over time. 

The second data generating process (DGP2) is identical to DGP1 except that the multivariate
marginals have $\etavec_t=(1,\ldots,1)^\top$ for all $t$. This corresponds to generating from
a stationary PVAR with $p=2$ lags.    
We generate 100 datasets from both DGP1 and DGP2. Each dataset has $n=500$ independent observations, where each observation $\xvec^{i}_{1:T}$ is generated sequentially from~\eqref{eq:simx}.

\subsection{Benchmarks and results}
We fit the following VD-vines with 2-GVC linking vector copulas and a PVAR benchmark to each simulated dataset:
\begin{itemize}
	\item \textbf{VDV}: YJN marginals {\em plus} stationary VD-vine copula
	\item \textbf{VDV(2):} Same as above, but where the VD-vine has known Markov order $p=2$
	\item \textbf{M1}: Multivariate Gaussian marginals (i.e. $\etavec_t=\bm{1}$) {\em plus} stationary VD-vine copula
		\item \textbf{M1(2)}: Same as above, but where the VD-vine has known Markov order $p=2$
	\item \textbf{M2}: Univariate YJN marginals (i.e. where $\Sigma_j=I_5$ for all $j$ in~\eqref{eq:Sj}) {\em plus} stationary D-vine copulas for each variable
		\item \textbf{M2(2)}: Same as above, but where the D-vines have known Markov order $p=2$
		\item \textbf{PVAR(2)}: A stationary Gaussian panel vector autoregression of Markov order 2
\end{itemize}
Model VDV is our proposed model with stationary assumptions on the copula unit hypercube domain, 
but not on the data domain, while VDV(2) is the same but where the correct
Markov order is imposed. M1 and M1(2) are sub-models used to assess if accounting for marginal asymmetry affects the accuracy of the fitted panel model. The sub-models M2 and M2(2) assume independence between the five variables and are included to assess if accounting for multivariate serial dependence affects the accuracy, while PVAR(2) is the benchmark. 
Table~\ref{tab:modprop} reports the number of parameters and summarizes key features of these models.

\begin{table}[htbp]
    \centering
    \caption{Features of the Different Models in the Simulation Study}
    \label{tab:sim_summary}
    \resizebox{\textwidth}{!}{%
    \begin{tabular}{lccccccc}
        \toprule
        \toprule
        Model
        & VDV & VDV(2) & M1 & M1(2) & M2 & M2(2) & PVAR(2) \\
        \midrule
        Number of Parameters
        & $(375,350)$
        & $(375,50)$
        & $(300,350)$
        & $(300,50)$
        & $(225,70)$
        & $(225,10)$
        & $70$ \\
        Asymmetric Marginals
        & $\checkmark$ & $\checkmark$ & & &
          $\checkmark$ & $\checkmark$ & \\
        Time-Varying Marginals
        & $\checkmark$ & $\checkmark$ & $\checkmark$ & $\checkmark$ &
		$\checkmark$ & $\checkmark$ & \\
        Nonlinear Dependence
        & $\checkmark$ & $\checkmark$ & &  &
          $\checkmark$ & $\checkmark$ & \\
        Multivariate Time Series
        & $\checkmark$ & $\checkmark$ & $\checkmark$ & $\checkmark$ &
          & & $\checkmark$ \\
        Markov
        &  & $\checkmark$ & & $\checkmark$ &
           & $\checkmark$ &$\checkmark$ \\
        \bottomrule
        \bottomrule
    \end{tabular}%
    }
    \label{tab:modprop}
    \caption*{Note: For each VD-vine model, the ordered pair gives
    $\lrp{\#\thetavec,\#\psivec}$, where $\thetavec$ are the marginal parameters and $\psivec$ are the stationary VD-vine parameters.  
    For PVAR(2), the reported value is the total number of parameters. A tick indicates that the corresponding feature is accommodated by the model.}
\end{table}
To evaluate the distributional fit we consider three scores for the 1-step ahead
in-sample predictive distributions. Let $F_{t|t-1}(\cdot|\xvec_{1:t-1})$ be the distribution function for the density at~\eqref{eq:TScondpdf}. Then for a scoring function ${\cal S}(F_{\mbox{\tiny pred}},\yvec)$ of distribution $F_{\mbox{\tiny pred}}$ evaluated at observed point $\yvec$ we compute the average score
\begin{equation}
\overline{{\cal S}}=\frac{1}{n} \sum_{i=1}^n\sum_{t=2}^T {\cal S}(F_{t|t-1}(\cdot|\xvec_{1:t-1}^i),\xvec^i_{t})\,.\label{eq:Sbar}
\end{equation}
We consider three different scoring functions that evaluate multivariate density forecast accuracy. The first is
the negative log score (LS), the second is a dimension-normalized energy score (ES) and the 
third is the squared maximum mean discrepancy (MMD); definitions of these scores are given in \ref{sm:evalfit} of the Online Appendix. Table~\ref{tab:simulation_scores} reports the average scores for the models under the two DGPs, and Figure~\ref{fig:simulation_boxplots} provides boxplots of the scores over the 100 replicates from DGP1.

\begin{table}[htbp]
    \centering
    \caption{Simulation Evaluation Scores under DGP1 and DGP2.}
    \label{tab:simulation_scores}
    \resizebox{\textwidth}{!}{%
    \begin{tabular}{lccc@{\hspace{2em}}lccc}
        \toprule
        \toprule
        \multicolumn{4}{c}{Panel A: DGP1} & \multicolumn{4}{c}{Panel B: DGP2} \\
        \cmidrule(lr){1-4}\cmidrule(lr){5-8}
        Model & Log Score & Energy Score & MMD & Model & Log Score & Energy Score & MMD \\
        \midrule
        VDV & 57.9586 & 7.7006 & 3.9366 & VDV & 57.0441 & 7.4110 & 3.9201 \\
        VDV(2) & \textbf{57.6641} & \textbf{7.6676} & \textbf{3.9192} & VDV(2) & 56.7448 & 7.3786 & 3.9023 \\
        M1 & 64.1007 & 7.7581 & 3.9739 & M1 & 56.9143 & 7.4087 & 3.9187 \\
        M1(2) & 63.8511 & 7.7261 & 3.9571 & M1(2) & 56.6339 & 7.3769 & 3.9012 \\
        M2 & 91.9623 & 8.7761 & 4.5210 & M2 & 91.0501 & 8.4610 & 4.5145 \\
        M2(2) & 92.2938 & 8.8195 & 4.5416 & M2(2) & 91.3809 & 8.5039 & 4.5357 \\
        PVAR(2) & 63.9500 & 7.7376 & 3.9584 & PVAR(2) & \textbf{56.2512} & \textbf{7.3597} & \textbf{3.8911} \\
        \bottomrule
        \bottomrule
    \end{tabular}%
    }
    \caption*{Note: Entries report the mean of $\overline{\mathcal S}$ over the 100 replications. Lower values indicate higher accuracy for all three scores, with the lowest value within each DGP in bold. The models are described in the text, with VDV(2) the correct model for DGP1 and PVAR(2) the correct model for DGP2.}
\end{table}

For DGP1, VDV and VDV(2) are the best performing because they nest the correct model. 
The sub-models M1 and M1(2) incorrectly enforce 
symmetry on the marginals, resulting in a reduction in accuracy.
A far greater reduction occurs when the multivariate aspect of the panel time series is ignored in M2 and M2(2). The benchmark PVAR(2) is similar to M1(2), except that it incorrectly enforces identical marginals over waves, resulting in a small reduction in accuracy relative to M1(2). 

For DGP2, the PVAR(2) is the correct model and results in the most accurate fit. However, the flexible VD-vine copula models VDV, VDV(2), M1 and M1(2) which nest the data generating process as a special case are only slightly less accurate. Taken together, the results of this simulation support the
importance of jointly modeling marginal asymmetry and multivariate serial
dependence.

\begin{figure}[thbp]
    \centering
    \includegraphics[width=\linewidth]{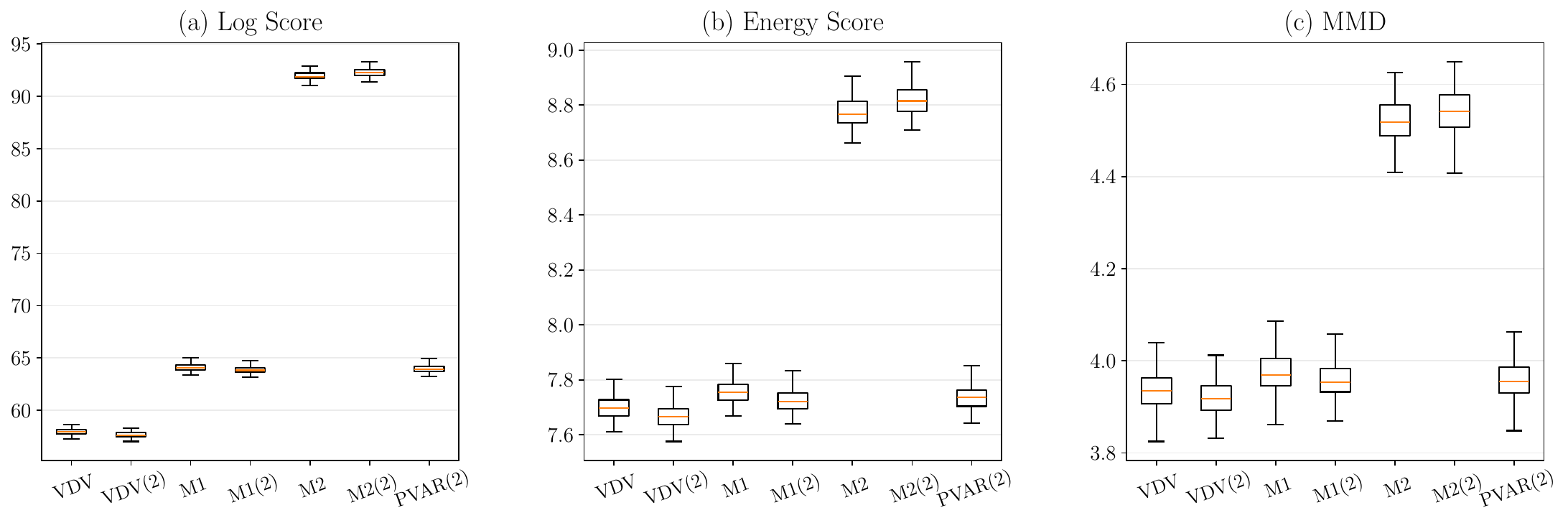}
    \caption{Boxplots of the log score, energy score,
    and squared maximum mean discrepancy across 100 replications for DGP1, with lower values indicating greater accuracy. Equivalent boxplots for 
    DGP2 are found in the Online Appendix.}
    \label{fig:simulation_boxplots}
\end{figure}

\section{Example: Work-Life Conflict and Well-Being}\label{sec:unbalanced}
Following~\cite{Yuetal2025} we study the temporal dependence between 
work-life conflict and subjective measures of well-being. These authors use 
the ``Household, Income and Labour Dynamics in Australia'' (HILDA) survey, which is 
a longitudinal household panel that collects detailed information on economic conditions, labor market outcomes, and individual well-being~~\citep{WatsonWooden2021}. As in~\cite{Yuetal2025} we construct a measure of work-life conflict (WLC), three measures of subjective well-being called Domain Satisfaction (DS), Positive Affect (PA), and reverse-coded Negative Affect (NA), and income (INC) as a control variable; see~\ref{sm:hildaextra} of the Online Appendix for details. We do this for
the \(n = 1093\) individuals observed in each of \(T = 8\) annual survey waves from 2016 to 2023. However, the multivariate outcome varies across waves because WLC is not collected in 2018, 2020 and 2022, so that $(d_1,\ldots,d_{8}) = (5, 5, 4, 5, 4, 5, 4, 5)$. 
 
In their original study~\cite{Yuetal2025} employ a dynamic structural equation model. Our objective is not to replicate their work, but to show how important distributional features that are not captured in their original study are captured by the proposed VD-vine copula model. These include (i)~multivariate asymmetry in the contemporaneous variable distributions, (ii)~nonlinear serial dependence that is not first order Markov, and (iii)~heterogeneity in the dependence structure over wave. We also stress that the original study did not model directly varying response vector composition, whereas the VD-vine does.  

Estimates of the YJN multivariate marginals specified in Section~\ref{sec:ltmaps} are far from Gaussian and exhibit substantial skewness. 
This is seen in Figure~\ref{fig:wave_1_marginals} which plots the 10 bivariate slices of the five-dimensional density for the 2016 wave (i.e. $t=1$).
These closely track the features in the underlying data, illustrating the flexibility of the transport at~\eqref{eq:Sj} in capturing the distribution. 
Table~\ref{tab:hilda_empirical_pair_skewness} reports Mardia's empirical  skewness~\citep{Mardia1970} for the bivariate marginals of all waves, where skew is both significant throughout and changes over waves.

\begin{figure}[htbp]
\centering
\caption{Contours of the fitted marginal density for the 2016 wave}
\resizebox{1.05\textwidth}{!}{%
\begin{tabular}{ccccc}
\includegraphics[width=0.2\linewidth]{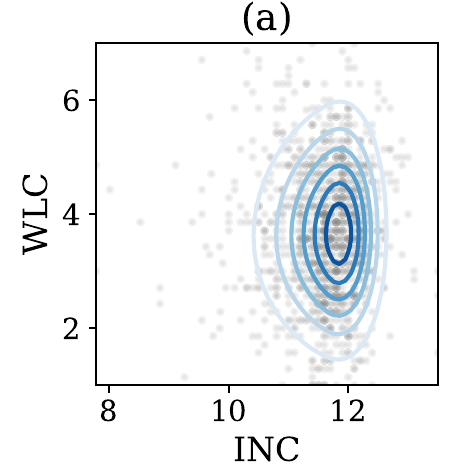} &
\includegraphics[width=0.2\linewidth]{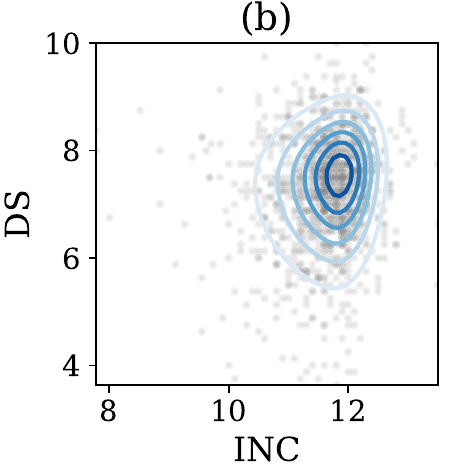} &
\includegraphics[width=0.2\linewidth]{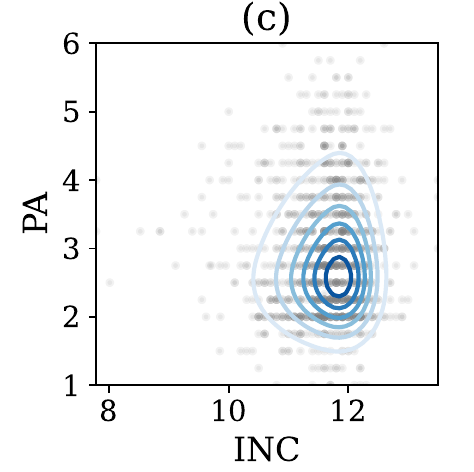} &
\includegraphics[width=0.2\linewidth]{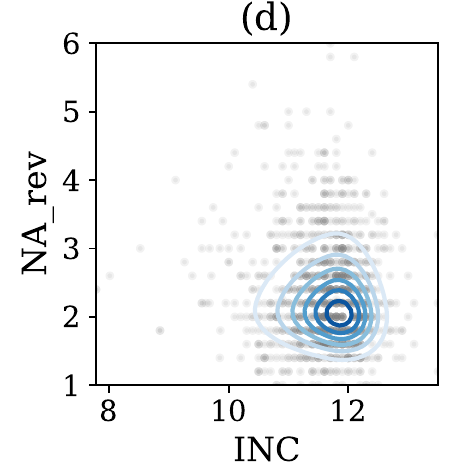} &
\includegraphics[width=0.2\linewidth]{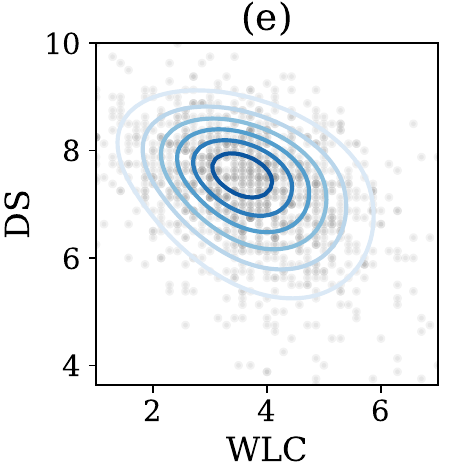}
\\
\includegraphics[width=0.2\linewidth]{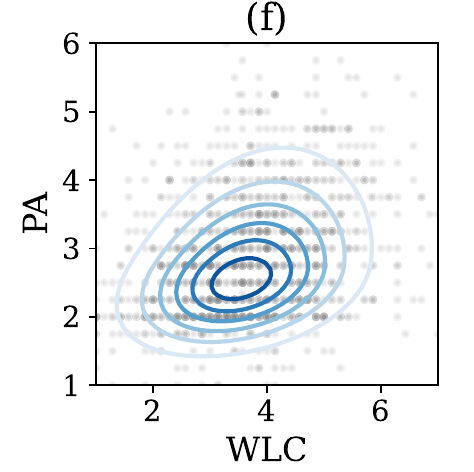} &
\includegraphics[width=0.2\linewidth]{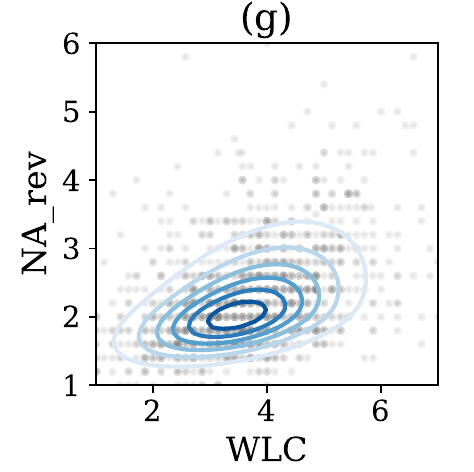} &
\includegraphics[width=0.2\linewidth]{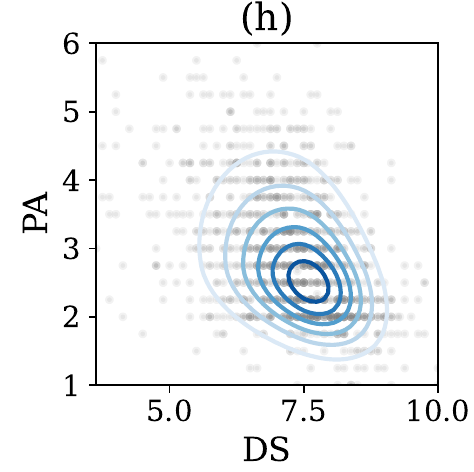} &
\includegraphics[width=0.2\linewidth]{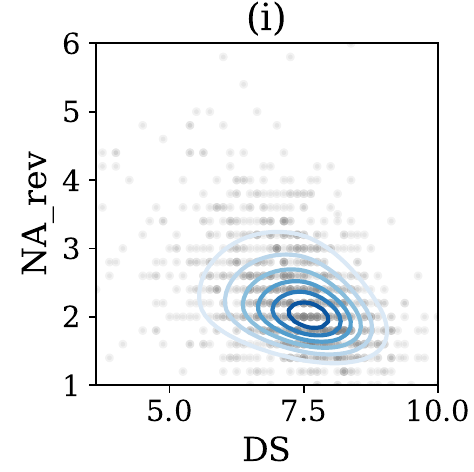} &
\includegraphics[width=0.2\linewidth]{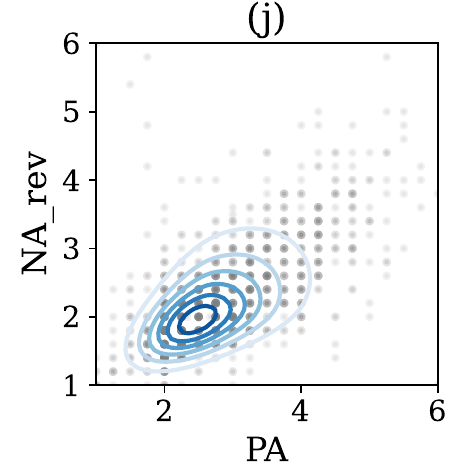}
\end{tabular}

}
\caption*{Note: Contours of all ten bivariate slices of the fitted $5$-dimensional marginal density for the 2016 wave (i.e. $t=1$). Also shown are scatterplots of the observed data.}
\label{fig:wave_1_marginals}
\end{figure}

\begin{table}[htbp]
\centering
\caption{Empirical Mardia skewness of the HILDA bivariate marginals.}
\label{tab:hilda_empirical_pair_skewness}
\resizebox{0.8\linewidth}{!}{
\begin{tabular}{lcccccccc}
\toprule
\toprule
Pair & 2016 & 2017 & 2018 & 2019 & 2020 & 2021 & 2022 & 2023 \\
\midrule
(INC, WLC) & 1.854 & 1.189 &       & 0.776 &       & 0.737 &       & 0.125 \\
(INC, DS)  & 2.202 & 1.546 & 1.704 & 1.069 & 0.908 & 1.178 & 1.564 & 0.410 \\
(INC, PA)  & 2.303 & 1.535 & 1.755 & 1.074 & 0.763 & 0.932 & 1.524 & 0.302 \\
(INC, NA)  & 2.884 & 1.946 & 2.509 & 1.718 & 1.335 & 1.645 & 2.290 & 1.061 \\
(WLC, DS)  & 0.309 & 0.291 &       & 0.359 &       & 0.410 &       & 0.423 \\
(WLC, PA)  & 0.483 & 0.404 &       & 0.356 &       & 0.257 &       & 0.272 \\
(WLC, NA)  & 1.179 & 0.812 &       & 0.990 &       & 0.968 &       & 1.225 \\
(DS, PA)   & 0.860 & 0.729 & 0.741 & 0.724 & 0.645 & 0.551 & 0.372 & 0.547 \\
(DS, NA)   & 1.441 & 1.049 & 1.452 & 1.207 & 1.094 & 1.214 & 1.121 & 1.330 \\
(PA, NA)   & 2.820 & 1.302 & 1.777 & 1.382 & 1.495 & 1.049 & 1.069 & 1.453 \\
\bottomrule
\bottomrule
\end{tabular}
}
\caption*{Note: Entries report Mardia's empirical multivariate skewness for each bivariate pair, with larger values indicating greater asymmetry. WLC was not collected in 2018, 2020, or 2022, so the corresponding cells are blank. Each available pair-year contains $n=1{,}093$ observations. Standard tests show that all skew coefficients are significantly different from zero.}
\end{table}

We fit the full VD-vine copula model and four sub-models below:
\begin{itemize}
	\item \textbf{VDV}: YJN marginals {\em plus} full VD-vine copula\,,
	\item \textbf{M1}: Multivariate Gaussian marginals (i.e. $\etavec=\bm{1}$) {\em plus} full VD-vine copula
	\item \textbf{M2}: Univariate YJN marginals (i.e. where $\Sigma_t=I_{d_t}$ for all $t$ in~\eqref{eq:Sj}) {\em plus} independent D-vine copulas for each variable
	\item \textbf{M3}: YJN marginals {\em plus} independence copula
	\item \textbf{M4}: YJN marginals {\em plus} Markov lag 1 (ie. $p=1$ in Section~\ref{sec:markovdef}) VD-vine copula
\end{itemize}
For these VD-vines (excluding M3) three linking vector copula choices are considered: (i)~2-GVC, (ii)~2-FGMVC and, (iii)~a hybrid case with 2-GVC for $C_{t,t-1}$ for $t=2,\ldots,T$ and 2-FGMVC for the rest.
The sub-models are selected to assess if marginal asymmetry (M1),
multivariate serial dependence (M2), any serial dependence (M3), and lag lengths longer than one (M4) are important features. As an additional benchmark we also employ stationary PVAR(1) and PVAR(2) models with Gaussian disturbances. 

Ten-fold cross-validation (CV) is used to assess predictive accuracy of the models.
For each fold, a model is estimated using 90\% of the individuals, and~\eqref{eq:Sbar} computed for the 10\% validation data.
Figure~\ref{fig:hilda_cv_energy_score_boxplot} gives boxplots of the mean scores over the 10 folds for the energy scoring function. The full VD-vine model with 2-GVC linking vector copulas (which has 777 parameters) is the most accurate,
illustrating the importance of the full set of dependence and marginal distribution features. The dependence structure is particularly important, with 
the M3 sub-models that omit this feature least accurate. Other observations include that (i)~the 2-GVC linking vector copulas outperform both the 2-FGMVC and Hybrid vector copulas, (ii)~the VDV outperforms M1 for each choice of vector linking copula indicating the importance of asymmetry in the marginals, (iii)~the PVAR(1) and PVAR(2) models are clearly dominated, and (iv)~imposing first order Markov dependence 
degrades accuracy. Figures~\ref{fig:hilda_cv_mmd_boxplot} and~\ref{fig:hilda_cv_logscore_boxplot} and Table~\ref{tab:hilda_cv_scores} in the Online Appendix report results for all scoring functions, which confirm these findings.

\begin{figure}[htbp]
    \centering
    \caption{Ten-fold cross-validation boxplots of the normalized energy score for the HILDA application.}
    \includegraphics[width=\textwidth]{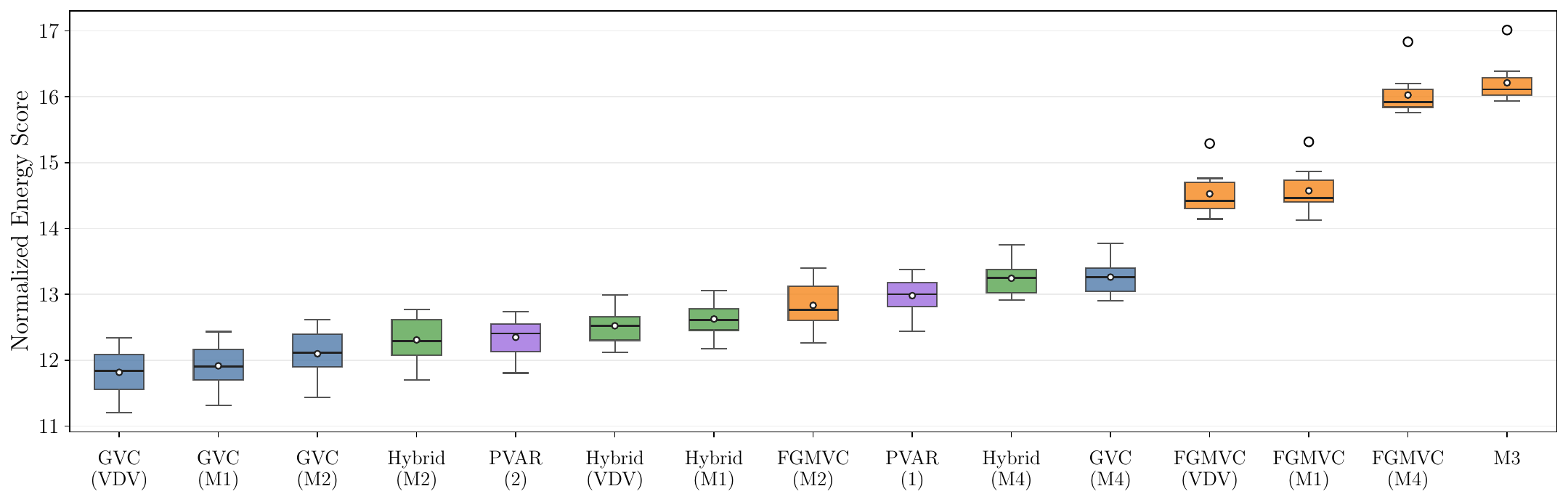}
    \label{fig:hilda_cv_energy_score_boxplot}
\end{figure}

Figure~\ref{fig:wlc_lagged_spearman} in the Online Appendix reports the Spearman correlations between WLC and well-being variables NA, PA, and DS from the fitted VD-vine. For each variable,  contemporaneous, one-year and two-year lagged correlation is plotted. WLC is positively associated with contemporaneous and lagged PA and NA, with correlations 
stable across waves. In contrast, WLC is negatively associated with DS, again with similar magnitudes across years and lags. The results are consistent with the findings of the original study that workplace conflict is strongly related with well-being.

  \section{Discussion}\label{sec:conc}
Conventional vine copulas have been very popular
in multivariate analysis because they are a modular way of defining flexible dependence structures in high dimensions~\citep{Czado2019}. Likelihood-based inference is 
made feasible in~\cite{aas2009pair} using the recursive identities in~\cite{joe1996}. 
A key contribution in our paper is to develop transport map analogues to these identities and show how to evaluate them for linking vector copula generalizations of the conventional pair-copulas. For the VD-vine, we use these to develop an efficient algorithm for the evaluation 
of its density, likelihood and predictive distributions. We establish that the VD-vine
is a valid vector copula, which is useful because the list of tractable vector copulas is currently very limited.

Our focus on the VD-vine is because of its applicability to modeling multivariate longitudinal (panel) data. It allows for 
different multivariate marginals for each wave, which is an important practical modeling consideration. This includes the case where the composition of variables differs over wave, which is commonplace~\citep{HoqueAcarTorabi2023}. 
The simulations show the efficacy of the methodology, with the VD-vine improving predictive accuracy when the marginals are asymmetric and the serial dependence is multivariate, while incurring little loss when the data are generated by a correctly specified Gaussian PVAR.
In the HILDA panel application, the full VD-vine with Gaussian linking vector copulas provides the most accurate cross-validated predictive distributions among the models considered. Comparisons with the nested models suggest that asymmetric margins, multivariate dependence, and dynamics beyond the first lag all contribute to accuracy. 
Some or all of these features are often overlooked in stochastic models for panel data.
We also stress that our use of VD-vine models for panel data is one of the first applications of vector
copulas. Related uses of copulas and transport maps with multivariate margins include~\cite{fan2023vector}, who model cross-sectional dependence in equity returns;~\cite{zhang2026copula}, who use cyclically monotone marginal transports in a Gaussian graphical model for multi-attribute data; and~\cite{fu2025vector}, who employ patterned vector copulas as fixed-form approximations in Bayesian VI.

We conclude with some suggestions for extensions of our work.
Richer linking vector copulas could accommodate stronger asymmetry and tail dependence, although constructing valid families with tractable forward and backward conditional transports remains a challenge. More expressive marginal transports could replace the YJN specification. Shrinkage priors could be employed for the parameters of the linking vector copulas to regularize the posterior in high dimensions. The same vector-node principle can also be applied to other vines, extending its use beyond longitudinal multivariate data. Finally, alternative posterior approximations and exact or hybrid simulation methods could be investigated when greater accuracy in posterior uncertainty is required. 
\noindent
\appendix
\section{2-Vector Copulas}\label{app:2vcop}
This appendix derives conditional transport maps and their 
associated inverses
for the 2-GVC and the 2-FGMVC vector copulas specified in Section~\ref{sec:linkvcs}.

\subsection{2-GVC}\label{app:2GVC}
Let $(\bm{V}_1^\top, \bm{V}_2^\top)^\top \sim C_v^{\text{GVC}}(\cdot; \Omega)$, where the vector copula is derived from a joint Gaussian distribution:
\[
\begin{bmatrix}
	\boldsymbol{Z}_1 \\
	\boldsymbol{Z}_2
\end{bmatrix}
\sim \mathcal{N}(\boldsymbol{0}, \Omega),
\quad \text{with} \quad
\Omega =
\begin{bmatrix}
	I_{d_1} & M^\top \\
	M & I_{d_2}
\end{bmatrix},
\]
with $M \in \bbR^{d_2 \times d_1}$ and $||M||_{op}<1$. Define $\boldsymbol{Z}_j = \underline{\Phi_1}^{-1}(\bm{V}_j)$ and $\zvec_j=\underline{\Phi_1}^{-1}(\vvec_j)$ for $j = 1,2$.
Then the conditional distribution satisfies
$\boldsymbol{Z}_1 | \boldsymbol{Z}_2 = \boldsymbol{z}_2 \sim \mathcal{N}(M^\top \boldsymbol{z}_2, \, \Sigma_{1|2})$
with $\Sigma_{1|2}= I_{d_1} - M^\top M$ positive definite. 
Denote the principal symmetric square root\footnote{Let the positive definite matrix $A$ have spectral decomposition $A=Q\Lambda Q^\top$, then the principal symmetric square root is $A^{1/2}=Q\Lambda^{1/2} Q^\top$ with $\Lambda^{1/2}$ denoting the diagonal matrix with the square roots of the elements of $\Lambda$.} of $\Sigma_{1|2}$ as $\Sigma_{1|2}^{1/2}$, and let $\widetilde{\bm Z}_j\sim \mathcal{N}(\bm{0},I_{d_j})$ for $j=1,2$. Adopting standard  notation, the linear map $T_{1|2}^Z(\vvec_1|\vvec_2)=\Sigma_{1|2}^{-1/2}(\zvec_1-M^\top \zvec_2)$ is the push-forward $T_{1|2}^Z(\cdot|\vvec_2)_{\#} \mathbb{P}_{{\bm Z}_1|{\bm Z}_2=\zvec_2}=\mathbb{P}_{\widetilde{\bm{Z}}_1}$. 
In addition, 
the push-forwards of the element-wise maps $\left(\underline{\Phi_1}^{-1}\right)_{\#}\mathbb{P}_{\bm{V}_1|\bm{V}_2=\vvec_2}= \mathbb{P}_{\bm{Z}_1|\bm{Z}_2=\zvec_2}$ and $\left(\underline{\Phi_1}\right)_{\#}\mathbb{P}_{\widetilde{\bm{Z}}_1}=\mathcal{U}_1$, where $\mathcal{U}_1$ is the uniform measure on $[0,1]^{d_1}$.
The composition of these three maps defines the forward conditional transport map:
\[
\vvec_{1|2}=H_{1|2}(\vvec_1 | \vvec_2) =\underline{\Phi_1}\circ T_{1|2}^Z(\cdot|\vvec_2)\circ\underline{\Phi_1}^{-1}(\vvec_1)= \underline{\Phi_{1}}\left(\Sigma_{1|2}^{-1/2}\left(\underline{\Phi_1}^{-1}(\vvec_1)-M^\top \boldsymbol{z}_2\right)\right),
\]
which is the required push-forward $(H_{1|2}(\cdot|\vvec_2))_{\#}\mathbb{P}_{\bm{V}_1|\bm{V}_2=\vvec_2}=\mathcal{U}_1$.
The
inverse conditional transport (pull-back transport) is
\[
\vvec_1=H_{1|2}^{-1}(\vvec_{1|2}|\vvec_2) = \underline{\Phi_1}(M^\top \boldsymbol{z}_2 + \Sigma_{1|2}^{1/2} \cdot \underline{\Phi_1}^{-1}(\vvec_{1|2}))\,.
\] 

Similarly,  
$\boldsymbol{Z}_2 | \boldsymbol{Z}_1 = \boldsymbol{z}_1 \sim \mathcal{N}(M \boldsymbol{z}_1, \,\Sigma_{2|1})\,,$ with 
$\Sigma_{2|1}=I_{d_2} - M M^\top$ positive definite. Let $\Sigma_{2|1}^{1/2}$ be the principal symmetric square root of $\Sigma_{2|1}$, then by a similar argument the backward conditional transport map is 
\[
\vvec_{2|1}=H_{2|1}(\vvec_2 | \vvec_1) = \underline\Phi_{1}\left(\Sigma_{2|1}^{-1/2}\left(\underline{\Phi_1}^{-1}(\vvec_2)- M\boldsymbol{z}_1\right)\right),
\]
with $\zvec_1=\underline{\Phi_1}^{-1}(\vvec_1)$. Its inverse conditional transport is
\[
H_{2|1}^{-1}\lrp{\vvec_{2|1}|\vvec_1}
=
\underline{\Phi}
\lrp{
	M\zvec_1
	+
	\Sigma_{2|1}^{1/2}
	\underline{\Phi}^{-1}\lrp{\vvec_{2|1}}
}, 
\]

We observe that $T^Z_{1|2}$ is a Brenier transport and is invariant to ordering of the vector. When conjugated by the marginal transports to the copula and original data scales to obtain $H_{1|2}$, it remains a valid transport.
The same arguments apply to $H_{2|1}$. 

\subsection{2-FGMVC}\label{app:2-fgmvc-kr}
Let $(\mV_1^\top,\mV_2^\top)^\top\sim C_{v}^{\mathrm{FGMVC}}(\cdot;M)$, where $M\in\bbR^{d_2\times d_1}$ and $\lVert M\rVert_{\mathrm{op}}<1/\sqrt{d_1d_2}$. Write $M=(M_1,\ldots,M_{d_1})$, where $M_j$ is the $j$th column of $M$, $\vvec_1=(v_{1,1},\ldots,v_{1,d_1})^\top$ and 
$\widetilde{\vvec}_r =\onevec_{d_r}-2\vvec_r$, for $r=1,2$. Because the marginal in $\mV_2$ is the product uniform, the density of $\mV_1|\mV_2=\vvec_2$ is
\begin{equation}
f_{{\bm V}_1|\bm{V}_2}(\vvec_1|\vvec_2)=c_{v}^{\mathrm{FGMVC}}(\vvec_1,\vvec_2;M)
=1+\sum_{j=1}^{d_1}M_j^\top\widetilde{\vvec}_2(1-2v_{1,j}),
\label{eq:fmgvc_cond1}
\end{equation}
where we write~\eqref{eq:fgmvc_dens} as a summation of terms in $M_j$ and \(v_{1,j}\).
The operator-norm restriction $\lVert M^\top\widetilde{\vvec}_2\rVert_1<1$ ensures the conditional density is strictly positive.

For the forward conditional transport map $H_{1|2}$ we use the Knothe--Rosenblatt (KR) transport, which is formed from successive scalar conditional distribution functions. For the coordinate order $1,\ldots,d_1$, let $\vvec_{1|2}=H_{1|2}(\vvec_1|\vvec_2)$ with $j$th component defined for $j=1,\ldots,d_1$ as
\[
v_{1\mid2,j}
:=\big[H_{1|2}(\vvec_1|\vvec_2)\big]_j
=F_{V_{1,j}\mid V_{1,1},\ldots,V_{1,j-1},\mV_2}
(v_{1,j}\mid v_{1,1},\ldots,v_{1,j-1},\vvec_2),
\]
where the conditioning variables preceding $V_{1,1}$ are omitted when $j=1$. For the FGMVC, these components and their inverses are available in closed form as shown below. From~\eqref{eq:fmgvc_cond1}, the density of the conditional distribution $(V_{1,j}| V_{1,1}=v_{1,1},\ldots,V_{1,j-1}=v_{1,j-1},\mV_2=\vvec_2)$ is 
\[
f(v_{1,j}|v_{1,1},\ldots,v_{1,j-1},\vvec_2)
=\frac{1+\sum_{k=1}^{j}M_k^\top\widetilde{\vvec}_2(1-2v_{1,k})}
{1+\sum_{k=1}^{j-1}M_k^\top\widetilde{\vvec}_2(1-2v_{1,k})}\,.
\]
Integrating this scalar conditional density from zero to $v_{1,j}$ yields
\[
v_{1\mid2,j}
=v_{1,j}+\frac{M_j^\top\widetilde{\vvec}_2}
{1+\sum_{k=1}^{j-1}M_k^\top\widetilde{\vvec}_2(1-2v_{1,k})}
v_{1,j}(1-v_{1,j}),
\qquad j=1,\ldots,d_1,
\]
where an empty sum is zero. The map $H_{1|2}(\cdot|\vvec_2)$ is lower triangular and strictly increasing in each current coordinate. 

The inverse map is evaluated in the same order. Once $v_{1,1},\ldots,v_{1,j-1}$ have been recovered, define the coefficient in the $j$th forward equation by
\[
a_{1\mid2,j}
:=\frac{M_j^\top\widetilde{\vvec}_2}
{1+\sum_{k=1}^{j-1}M_k^\top\widetilde{\vvec}_2(1-2v_{1,k})}.
\]
The $j$th component of the inverse conditional transport map is
\[
\lrb{H_{1|2}^{-1}(\vvec_{1|2}|\vvec_2)}_j
=v_{1,j}
=\frac{2v_{1|2,j}}
{1+a_{1|2,j}+\sqrt{(1+a_{1|2,j})^2-4a_{1|2,j}v_{1|2,j}}},
\]
for $j=1,\ldots,d_1$. This sequential recursion defines $\vvec_1=H_{1|2}^{-1}(\vvec_{1|2}|\vvec_2)$.

The backward conditional transport $H_{2|1}(\vvec_2|\vvec_1)$ and its inverse follow from the same construction as the forward conditional transport by interchanging the subscripts $1$ and $2$, replacing $d_1$ by $d_2$, and replacing $M_j^\top\widetilde{\vvec}_2$ by $(M\widetilde{\vvec}_1)_j$. 

Unlike the symmetric Gaussian maps in Appendix~\ref{app:2GVC}, the KR maps are triangular and depend on the chosen coordinate order.  Nevertheless, they satisfy the bijective push-forward condition in Assumption~\ref{asmp:h_fun}. The natural order declared above is therefore fixed throughout estimation and forecasting.

  \newpage
  \singlespacing
  \bibliography{references}
\newpage
\noindent
\begin{center}
	{\bf \Large{Online Appendix for ``Vector Vine Copula Models for Multivariate Longitudinal Data''}}
\end{center}
\spacing{1.5}
\vspace{10pt}
\setcounter{page}{1}
\setcounter{figure}{0}
\setcounter{table}{0}
\setcounter{section}{0}
\setcounter{equation}{0}
\setcounter{algorithm}{1}
\renewcommand{\thetable}{A\arabic{table}}
\renewcommand{\thefigure}{A\arabic{figure}}
\renewcommand{\thealgorithm}{\Alph{section}\arabic{algorithm}}
\renewcommand{\thesection}{Part~\Alph{section}}
\renewcommand{\theequation}{A\arabic{equation}}
\makeatletter
  \renewcommand\@seccntformat[1]{\csname the#1\endcsname\quad}
\makeatother

\noindent

This Online Appendix has six parts:
\begin{itemize}
    \item[] {\bf Part~A}: Proofs
    \item[] {\bf Part~B}: Simulation Details
    \item[] {\bf Part~C}: Details on the Evaluation of Conditional Distributional Fit
    \item[] {\bf Part~D}: Benchmark PVAR Model for Missing Data
    \item[] {\bf Part~E}: Further Information on the HILDA Panel Data Example
    \item[] {\bf Part~F}: Extra Details on the Hyper-Spherical Parameterization
\end{itemize}
\newpage

\section{Proofs}\label{sm:grs}
This section contains the proofs of all lemmas and Theorem~1.

\begin{proof}[{\bf Proof of Lemma~\ref{lem:cvdag}}]
	For \(c^\dagger\) to be a \(T\)-vector copula density as specified in Definition~\ref{def:kvcop}, the following have to hold:
	\begin{itemize}
	\item[(i)] Non-negativity: Since \(c_{t,j}\geq 0\) for all \(t,j\), it follows from~\eqref{eq:cvdagger} that \(c^\dagger(\uvec_{1:T})\geq 0\) on \((0,1)^d\).	
	\item[(ii)] Normalization: Proceed recursively in \(t\). For \(t=1\), \(f(\uvec_1)=1\). Suppose that \(f(\uvec_{1:t-1})\) is a normalized density. Starting from the product-uniform marginal \(f(\uvec_t)=1\), apply the conditional vector-copula construction at~\eqref{eq:u2vcop} successively for \(j=t-1,\ldots,1\). At each step,~\eqref{eq:ucond} defines a normalized conditional density, and hence the density at~\eqref{eq:ucondpdf} satisfies
	$$
	\int_{(0,1)^{d_t}} f(\uvec_t\mid \uvec_{1:t-1})\,d\uvec_t=1.
	$$
	Therefore, \(f(\uvec_{1:t})=f(\uvec_{1:t-1})f(\uvec_t\mid \uvec_{1:t-1})\) is normalized. Repeating this argument to \(t=T\) gives
	$$
	\int_{(0,1)^d}c^\dagger(\uvec_{1:T})\,d\uvec_{1:T}=1.
	$$
	\item[(iii)] Marginal Uniformity: Integrating the conditional joint density at~\eqref{eq:u2vcop} over \(\uvec_j\) gives
	$$
	\int f(\uvec_t\mid \uvec_{j:t-1})
	f(\uvec_j\mid \uvec_{j+1:t-1})\,d\uvec_j
	=
	f(\uvec_t\mid \uvec_{j+1:t-1}).
	$$
	Applying this identity successively for \(j=1,\ldots,t-1\) yields
	$$
	\int f(\uvec_t\mid \uvec_{1:t-1})f(\uvec_{1:t-1})\,d\uvec_{1:t-1}
	=f(\uvec_t)=1.
	$$
	Thus, \(\bm{U}_t\) has a product-uniform marginal. The marginals of \(\bm{U}_1,\ldots,\bm{U}_{t-1}\) are unchanged because \(f(\uvec_t\mid \uvec_{1:t-1})\) integrates to one. Hence, by induction, every block marginal of \(c^\dagger\) is product uniform, completing the proof.
\end{itemize}
\end{proof}

\begin{proof}[{\bf Proof of Theorem~\ref{thm:recursive}}]

Fix \(t=2,\ldots,T\) and \(j<t\), and condition throughout on
$\mU_{j+1:t-1}=\uvec_{j+1:t-1}$,
	with the convention that this conditioning set is empty when \(j=t-1\).
	The dependence of the conditional maps on \(\uvec_{j+1:t-1}\) is suppressed.
	
Define
\[
\mV_t=R^{-1}_{t|j+1}(\mU_t),
\qquad
\mV_j=R^{-1}_{j|t-1}(\mU_j).
\]
By~\eqref{eq:u2vcop}, this is the \(k=2\) vector copula model of Definition~\ref{def:vcopmod} applied to the conditional distribution of
\((\mU_t,\mU_j)\) given \(\mU_{j+1:t-1}=\uvec_{j+1:t-1}\). Hence, conditionally on
\(\mU_{j+1:t-1} = \uvec_{j+1:t-1}\), 
the pair
\((\mV_t,\mV_j)\) has linking vector copula \(C_{t,j}\), density \(c_{t,j}\), and marginal distributions \(\calU_t\) and \(\calU_j\).
	
We first prove the forward recursion. Fix \(\uvec_j\) and define the vector value
$\uvec_{j|t-1}:=R^{-1}_{j|t-1}(\uvec_j)$. 
Because \(R^{-1}_{j|t-1}\) is bijective, conditioning on \(\mU_j=\uvec_j\), in addition to
\(\mU_{j+1:t-1}=\uvec_{j+1:t-1}\), 
is equivalent to conditioning on
\(\mV_j=\uvec_{j|t-1}\).
Therefore,
\(
R^{-1}_{t|j+1}(\mU_t) \mid \mU_{j:t-1}=\uvec_{j:t-1}
\)
has distribution 
\(
\bbP_{\mV_t\mid \mV_j=\uvec_{j|t-1}}
\)
under the linking vector copula \(C_{t,j}\). By Assumption~\ref{asmp:h_fun},
\[
\lrb{H_{t|j}(\cdot\mid\uvec_{j|t-1})}_{\#}
\bbP_{\mV_t\mid \mV_j=\uvec_{j|t-1}} = \calU_t .
\]
Combining these two transformations gives the composition function
\[
\lrb{ H_{t|j}(\cdot\mid\uvec_{j|t-1})\circ R^{-1}_{t|j+1}}_{\#} \bbP_{\mU_t\mid \mU_{j:t-1}=\uvec_{j:t-1}} 
= \calU_t .
\]
which is the pull-back from
\(\mU_t\mid\mU_{j:t-1}=\uvec_{j:t-1}\) to \(\calU_t\).
With the pull-backs \(R^{-1}_{s|r}\) understood as those generated by the VDV construction, this composition function equals \(R^{-1}_{t|j}\). Hence, for any \(\uvec_t\),
\[
R^{-1}_{t|j}(\uvec_t) =
H_{t|j} \lrp{ R^{-1}_{t|j+1}(\uvec_t)\mid \uvec_{j|t-1}}.
\]
Using \(\uvec_{s|r} \coloneq R^{-1}_{s|r}(\uvec_s)\), 
we obtain
\( \uvec_{t|j} = H_{t|j}(\uvec_{t|j+1}\mid \uvec_{j|t-1})\).
	
The backward recursion is analogous. Fix \(\uvec_t\) and define the vector value 
\(\uvec_{t|j+1}:=R^{-1}_{t|j+1}(\uvec_t)\).
Because \(R^{-1}_{t|j+1}\) is bijective, conditioning on
\(\mU_t=\uvec_t\), in addition to
\(\mU_{j+1:t-1}=\uvec_{j+1:t-1}\), is equivalent to conditioning on \(\mV_t=\uvec_{t|j+1}\).
Therefore
\(
R^{-1}_{j|t-1}(\mU_j) \mid \mU_{j+1:t}=\uvec_{j+1:t}
\)
has distribution
\( \bbP_{\mV_j\mid \mV_t=\uvec_{t|j+1}} \)
under \(C_{t,j}\). 
By Assumption~\ref{asmp:h_fun},
\[
\lrb{H_{j|t}(\cdot\mid\uvec_{t|j+1})}_{\#} \bbP_{\mV_j\mid \mV_t=\uvec_{t|j+1}} = \calU_j .
\]
Combining these two transformations gives the composition
\[
\lrb{ H_{j|t}(\cdot\mid\uvec_{t|j+1})\circ R^{-1}_{j|t-1} }_{\#} \bbP_{\mU_j \mid \mU_{j+1:t}=\uvec_{j+1:t}} 
=\calU_j,
\]
which is the conditional pull-back \(R^{-1}_{j|t}\), and so
\[
R^{-1}_{j|t}(\uvec_j) = 
H_{j|t}\!\lrp{ R^{-1}_{j|t-1}(\uvec_j)\mid \uvec_{t|j+1}},
\]
equivalently,
\( \uvec_{j|t} = H_{j|t}(\uvec_{j|t-1}\mid \uvec_{t|j+1}) \).
When \(j=t-1\), the conditioning set is empty and the stated conventions \(\uvec_{t|t}=\uvec_t\) and \(\uvec_{t-1|t-1}=\uvec_{t-1}\) give the adjacent-block case. This proves both recursions.
\end{proof}

\begin{proof}[{\bf Proof of Lemma~\ref{lem:free_param}}]
	Let \(M\in{\cal M}_\delta\), and write a singular value decomposition as
	\(
	M=Q_2\Sigma Q_1^\top ,
	\)
	where \(Q_2\) and \(Q_1\) are orthogonal matrices and \(\Sigma\) is diagonal 
	matrix with singular values \(\sigma_1,\ldots,\sigma_r\), where
	\(r=\min(d_1,d_2)\). Since \(\|M\|_{\mathrm{op}}<\delta\), we have
	\(0\leq\sigma_i<\delta\) for all \(i\), so \(I_{d_1}-\delta^{-2}M^\top M\) is positive
	definite. Therefore
	\[ 
	\Psi_\delta(M)
	=
	Q_2\delta^{-1}\Sigma(I_{d_1}-\delta^{-2}\Sigma^\top\Sigma)^{-1}Q_1^\top .
	\]
	Thus \(\Psi_\delta(M)\) has the same left and right singular vectors as \(M\), and
	its singular values are
	\[
	\beta_i=\frac{\sigma_i/\delta}{1-\sigma_i^2/\delta^2},\qquad i=1,\ldots,r .
	\]
	The scalar map
	\[
	\sigma\mapsto \frac{\sigma/\delta}{1-\sigma^2/\delta^2},\qquad 0\leq\sigma<\delta,
	\]
	is strictly increasing and maps \([0,\delta)\) onto \([0,\infty)\). Its inverse is
	\[
	\sigma
	=
	\delta\frac{2\beta}{1+\sqrt{1+4\beta^2}},
	\]
	with the value at \(\beta=0\) understood by continuity.
	
	Now take any \(B\in\mathbb{R}^{d_2\times d_1}\), and write its singular value
	decomposition as
	\(
	B=Q_2\Gamma Q_1^\top
	\),
	with singular values \(\beta_1,\ldots,\beta_r\). Define
	\[
	\sigma_i=\delta\frac{2\beta_i}{1+\sqrt{1+4\beta_i^2}},
	\qquad i=1,\ldots,r,
	\]
	and set \(M=Q_2\Sigma Q_1^\top\), where \(\Sigma\) is rectangular diagonal
	with entries \(\sigma_i\). Since \(0\leq\sigma_i<\delta\), this \(M\) belongs to
	\({\cal M}_\delta\), and the preceding calculation gives \(\Psi_\delta(M)=B\). Hence
	\(\Psi_\delta\) is onto. The same scalar monotonicity also gives uniqueness of the
	singular values, with the singular subspaces inherited from \(B\), so
	\(\Psi_\delta\) is one-to-one.
	
	It remains only to verify the stated matrix form of the inverse. Since
	\( 
	B^\top B =
	Q_1\Gamma^\top\Gamma Q_1^\top ,
	\)
	the principal square root satisfies
	\[
	(I_{d_1}+4B^\top B)^{1/2}
	=
	Q_1(I_{d_1}+4\Gamma^\top\Gamma)^{1/2}Q_1^\top .
	\]
	Therefore
	\[
	2\delta B\bigl(I_{d_1}+(I_{d_1}+4B^\top B)^{1/2}\bigr)^{-1}
	=
	Q_2\Sigma Q_1^\top
	=
	M,
	\]
	because the singular values on the right-hand side are exactly
	\[
	\sigma_i=\delta\frac{2\beta_i}{1+\sqrt{1+4\beta_i^2}} .
	\]
	Since each \(\sigma_i<\delta\), the inverse transformation satisfies
	\(\|M\|_{\mathrm{op}}<\delta\).
\end{proof}

\begin{proof}[{\bf Proof of Lemma~\ref{lem:compforwardmap}}]		
For \(t=1\), the first identity is the empty-composition
case and the second follows from
\(
(S_1^{-1})_{\#}P_{\mX_1}=\calU_1 .
\)
Now suppose \(t>1\), and condition on
\(\mU_{1:t-1}=\uvec_{1:t-1}\). Throughout, write
\(
H_{t|j}(\cdot) \coloneq
H_{t|j}(\cdot\mid \uvec_{j|t-1}),
\)
for \(j=1,\ldots,t-1\).

By Theorem~\ref{thm:recursive},
\(
\uvec_{t|j}=H_{t|j}(\uvec_{t|j+1}),
\)
for \( j=t-1,\ldots,1 \), 
with \(\uvec_{t|t}\equiv \uvec_t\). Repeated substitution therefore gives
\[
\uvec_{t|1}
=
\left(H_{t|1}\circ H_{t|2}\circ\cdots\circ H_{t|t-1}\right)(\uvec_t).
\]
Since \(\uvec_{t|1}=R^{-1}_{t|1}(\uvec_t)\), it follows that
\(
R^{-1}_{t|1} =
H_{t|1}\circ H_{t|2}\circ\cdots\circ H_{t|t-1}.
\)
By definition, \(R_{t|1}\) pushes \(\calU_t\) forward to
\(P_{\mU_t\mid \mU_{1:t-1}=\uvec_{1:t-1}}\). Hence
\[
\left(
H_{t|1}\circ H_{t|2}\circ\cdots\circ H_{t|t-1}
\right)_{\#}
P_{\mU_t\mid \mU_{1:t-1}=\uvec_{1:t-1}}
=
\calU_t ,
\]
which proves~\eqref{eq:ctransportu}.

For~\eqref{eq:conditionaltransport}, set
\(
\uvec_s=S_s^{-1}(\xvec_s),
\)
for \(s=1,\ldots,t-1\).
Because the marginal maps are bijective, conditioning on
\(\mX_{1:t-1}=\xvec_{1:t-1}\) is equivalent to conditioning on
\(\mU_{1:t-1}=\uvec_{1:t-1}\), and
\[
(S_t^{-1})_{\#}
P_{\mX_t\mid \mX_{1:t-1}=\xvec_{1:t-1}} 
=
P_{\mU_t\mid \mU_{1:t-1}=\uvec_{1:t-1}} .
\]
Composing this identity with~\eqref{eq:ctransportu} gives
\[
\left(
H_{t|1}\circ H_{t|2}\circ\cdots\circ H_{t|t-1}\circ S_t^{-1}
\right)_{\#}
P_{\mX_t\mid \mX_{1:t-1}=\xvec_{1:t-1}}
=
\calU_t ,
\]
which proves the result.
\end{proof}

\begin{proof}[{\bf Proof of Lemma~\ref{lem:station}}]
Fix $m\geq 0$ and an admissible starting point $s$. Successive
integration of the endpoint blocks in~\eqref{eq:cvdagger} shows that the marginal density of
$\bm{U}_{s:s+m}$ is
\[
c^\dagger_{s:s+m}(\bm{u}_{s:s+m})
=
\prod_{t=s+1}^{s+m}\prod_{j=s}^{t-1}
c_{t,j}\bigl(\bm{u}_{t\mid j+1},\bm{u}_{j\mid t-1}\bigr),
\]
where the arguments are generated recursively using the orientated linking
specifications $L_{t,j}$. At each integration, the factors involving the
removed endpoint block form a normalized conditional density and therefore
integrate to one.

If ${\cal L}_{t,t-\ell}={\cal L}_{t',t'-\ell}$ for every lag $\ell$, then re-labeling the
indices by $t\mapsto t-s+1$ and $j\mapsto j-s+1$ leaves both the linking
copulas and the recursively generated arguments unchanged. Hence,
\[
(\bm{U}_s,\ldots,\bm{U}_{s+m})
\stackrel{d}{=}
(\bm{U}_{s'},\ldots,\bm{U}_{s'+m})
\]
for any two admissible starting points $s$ and $s'$. Since arbitrary finite
sub-vectors are marginals of consecutive blocks, $\{\bm{U}_t\}$ is strictly
stationary.

If, in addition, $S_t=S$ for all $t$, then
\[
(\bm{X}_s,\ldots,\bm{X}_{s+m})
=
\bigl(S(\bm{U}_s),\ldots,S(\bm{U}_{s+m})\bigr).
\]
Applying the same measurable transformation component-wise to
equal-in-distribution vectors establishes strict stationarity of
$\{\bm{X}_t\}$.
\end{proof}

\newpage
\section{Simulation Details}\label{sm:sim}

We conduct \(N=100\) Monte Carlo replications. In each replication, we
generate \(n=500\) independent samples for
estimation. Each replication consists of \(T=15\) vectors of dimension \(\tilde{d}=5\).

The two DGPs are designed to share exactly the same latent Gaussian dependence. 
Let $d=T\tilde{d}$, 
\(\mZ = \lrp{\mZ_{1}^{\top},\ldots,\mZ_{T}^{\top}}^{\top}
\sim \calN_{d}\lrp{\zerovec,\mR},
\) and \( \Var\lrp{\mZ_{t}}=I_{\tilde d}. \)
Define the lag-\(h\) correlation block by
\( R_h = \Cov\lrp{\mZ_{t},\mZ_{t-h}}, \) with \( R_0=I_{\tilde d}\). 
The joint correlation matrix is therefore
\[
\mR=
\begin{pmatrix}
R_0 &     &  &   & \\
R_1 & R_0 &  &   & \\
R_2 & R_1 &  R_0 & & \\
\vdots&\vdots&\vdots&\ddots& \\
R_{T-1} & R_{T-2} & R_{T-3} & \cdots & R_0
\end{pmatrix}.
\]
Thus, stationarity is equivalent to requiring that each covariance block
depends only on the lag \(h\), making \(\mR\) block Toeplitz.

\subsection{DGP1: Stationary VD-Vine and Marginals}

We construct \(\mR\) from a Gaussian VD-vine of Markov order two. Following the \(M\) parameterization, the first two layers of the VD-vine are specified by the \(\tilde{d}\times \tilde{d}\) matrices \(M_1\) and \(M_2\) for linking copulas \(C_{t,t-1} \) and \( C_{t,t-2\mid t-1} \) respectively,
while all higher-order copulas, \( C_{t,t-h\mid t-h+1:t-1} \) with \( h>2 \), are independent. The fixed VD-vine parameters are
\[
M_1=
\begin{pmatrix}
-0.0651& 0.4400& 0.0057&-0.1582&-0.1012\\
-0.0485&-0.2213&-0.0976&-0.3525&-0.0929\\
-0.2817&-0.0327& 0.3167&-0.0211& 0.1420\\
-0.3066&-0.1161&-0.3081& 0.1116& 0.0281\\
-0.0633&-0.1006& 0.0523& 0.1439&-0.4309
\end{pmatrix},
\]
\[
M_2=
\begin{pmatrix}
-0.1297& 0.0881&-0.1347&-0.0933&-0.0891\\
 0.0584&-0.1461& 0.0068&-0.1260& 0.0224\\
 0.1555& 0.2203& 0.0637&-0.0221& 0.0262\\
 0.0925&-0.0268&-0.1518& 0.1665&-0.1121\\
 0.0352&-0.0146&-0.1491&-0.0210& 0.1492
\end{pmatrix},
\]
and \(M_h = 0\) for \(h > 2\).



The first two lag matrices are determined by the nonzero VD-vine parameters:
\( R_1=M_1 \),  and
\[
R_2
=
R_1R_1+
\lrp{I_{\tilde d}-R_1R_1^\top}^{1/2}
M_2
\lrp{I_{\tilde d}-R_1^\top R_1}^{1/2}.
\]
All higher-tree VD-vine parameters are set to zero. Consequently, for
\(h>2\), the remaining lag correlations are obtained recursively as
\[
R_h
=
\begin{pmatrix}
R_1 & R_2 & \cdots & R_{h-1}
\end{pmatrix}
\begin{pmatrix}
I_{\tilde d}          & R_1          & \cdots & R_{h-2}\\
R_1^\top     & I_{\tilde d}          & \cdots & R_{h-3}\\
\vdots       & \vdots       & \ddots & \vdots\\
R_{h-2}^\top & R_{h-3}^\top & \cdots & I_{\tilde d}
\end{pmatrix}^{-1}
\begin{pmatrix}
R_{h-1}\\
R_{h-2}\\
\vdots\\
R_1
\end{pmatrix}.
\]
Since the expression for \(R_h\) involves only the previously recovered
matrices \(R_1,\ldots,R_{h-1}\), the two VD-vine parameters \(M_1\) and
\(M_2\) determine the complete lag-correlation sequence. 
Figure~\ref{fig:sim_dependence} presents the resulting correlation and standardized precision matrices. 
\begin{figure}[htbp]
    \centering
    \includegraphics[width=\textwidth]
    {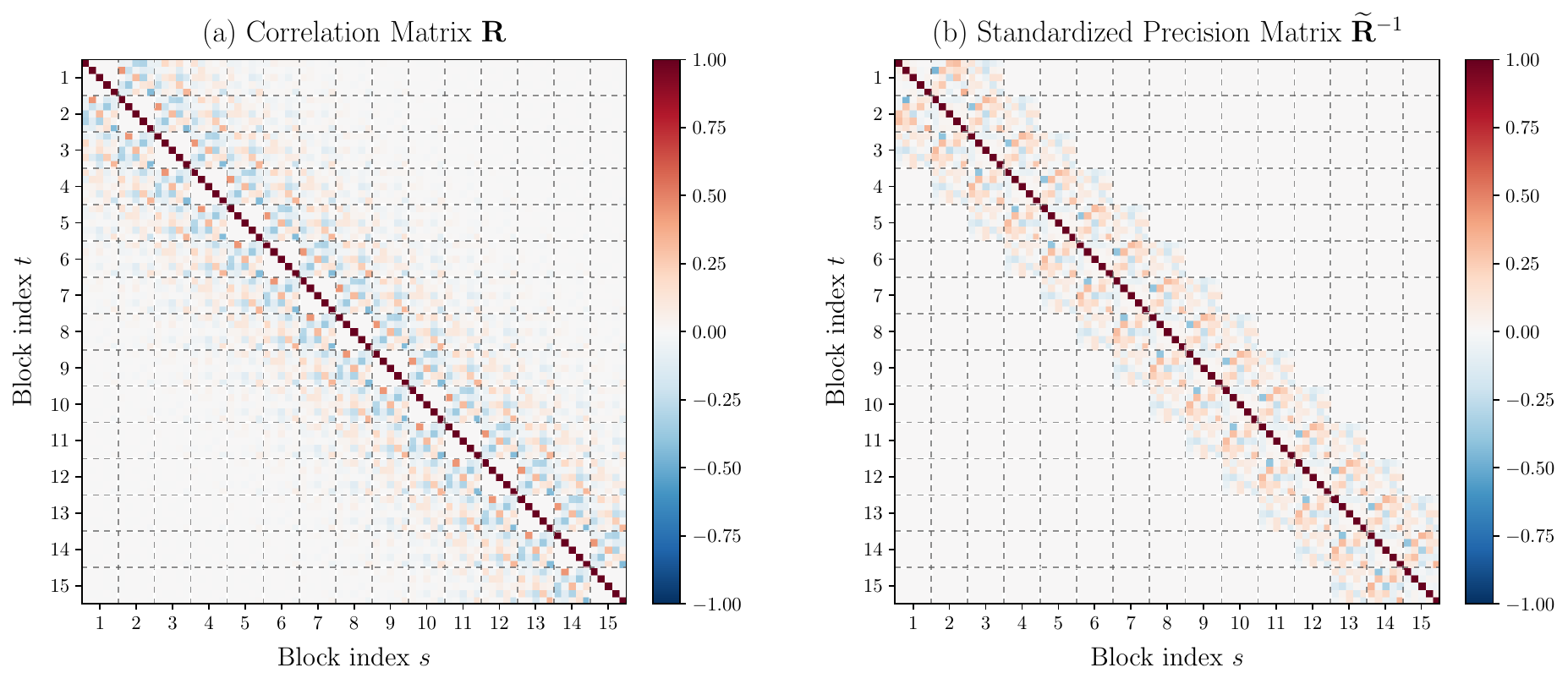}
    \caption{Common latent Gaussian dependence matrix used in both DGPs.
    Panel (a) presents the block correlation matrix \(\mR\), and panel
    (b) presents its standardized precision matrix. The repeated block
    diagonals of \(\mR\) reflect stationarity, while the block-banded
    precision matrix reflects the Markov-order-two restriction. Dashed lines
    identify the \(T=15\) vector blocks of dimension \(\tilde{d}=5\).}
    \label{fig:sim_dependence}
\end{figure}

DGP1 combines the stationary Gaussian VD-vine dependence described above with time-varying multivariate YJN margins. 
For each \(t\), let \(\muvec_t\), \(D_t\), and \(\etavec_t\) denote,
respectively, the location vector, diagonal scale matrix, and vector of
Yeo--Johnson parameters. The multivariate YJN construction is
\[
k_{\etavec_t} \lrp{ D_t^{-1}\lrp{\mX_t-\muvec_t} } = L\mZ_t,
\]
where \(k_{\etavec_t}\) applies the Yeo--Johnson transformation
componentwise. 
We fix the location and scale parameters at
\( \muvec_t=\zerovec_{\tilde d}\), \(D_t=I_{\tilde d}\) for \(t=1,\ldots,T\).
Let \(\Sigma \) denote the
within-vector correlation matrix and let \(L\) be its lower-triangular
Cholesky factor, \( \Sigma = LL^\top \).
The same \(\Sigma\) is used at every time point, specifically, we set
\[
\Sigma =
\begin{pmatrix}
 1       &  & &  &  \\
 0.1418  & 1      &  & & \\
 0.8868  & 0.3425 & 1      & &  \\
 0.3299  & 0.2043 & 0.2831 & 1      &  \\
-0.6624  &-0.1600 &-0.3459 &-0.1324 & 1
\end{pmatrix}.
\]
We use the time-varying asymmetry parameters
\[
\etavec_t=
\begin{cases}
\lrp{1,1,1,1,1}^{\top}, & t=1,\\
\lrp{0.60,1.40,0.70,1.30,1}^{\top}, & t\text{ even},\\
\lrp{1.35,0.65,1.30,0.70,1}^{\top}, & t\geq3\text{ odd}.
\end{cases}
\] 
The first vector margin is therefore symmetric. The remaining vector margins produce asymmetry in opposite directions across coordinates.
The within-vector Gaussian dependence is common across time and both DGPs. 
Here, \( L \) is the lower-triangular Cholesky factor used only for generating and estimating the \( \Sigma \). In transport from the fitted YJN marginal to the vector-copula scale uses the principal symmetric inverse square root \( \Sigma^{-1/2} \). 

The serial copula of DGP1 remains stationary because all latent vectors are
generated from the common block-Toeplitz correlation matrix \(\mR\).
However, the observed process is not marginally stationary because
\(\etavec_t\) varies with \(t\). DGP1 therefore evaluates whether the
VD-vine model can recover stationary serial dependence in the presence of
time-varying asymmetric margins.

\subsection{DGP2: Stationary Gaussian PVAR Model}

The construction of the matrix \(\mR\) for latent random vectors
\(\lrc{\mZ_t}\) uses the \(\lrc{M_h}\) parameterization as partial
correlation matrices. Moreover, \(M_h=\bm{0}\) for \(h>2\) implies that
the latent process is Markov of order two, which can be expressed as
\(\mZ_t \perp \lrp{\mZ_{t-3},\mZ_{t-4},\ldots}
\mid \lrp{\mZ_{t-1},\mZ_{t-2}}\). It remains to show that this stationary
Gaussian Markov process has a VAR(2) representation.

Under the definition \(R_h=\Cov\lrp{\mZ_t,\mZ_{t-h}}\), the covariance
matrix of the three vectors in chronological order is
\[
\Var\lrp{
\begin{pmatrix}
    \mZ_{t-2}\\
    \mZ_{t-1}\\
    \mZ_t
\end{pmatrix}
}
=
\begin{pmatrix}
I_{\tilde d}&R_1^\top&R_2^\top\\
R_1&I_{\tilde d}&R_1^\top\\
R_2&R_1&I_{\tilde d}
\end{pmatrix}.
\]
Thus, \(R_h\) appears below the block diagonal and \(R_h^\top\) appears
above it.

Define
\(\mG_t=\lrp{\mZ_{t-1}^\top,\mZ_{t-2}^\top}^\top\). After reordering the
same three vectors so that the current vector appears first, the
covariance matrix of the stacked vector
\(\lrp{\mZ_t^\top,\mG_t^\top}^\top\) is
\[
\Var\lrp{
\begin{pmatrix}
    \mZ_t\\
    \mG_t
\end{pmatrix}
}
=
\begin{pmatrix}
I_{\tilde d}&\Sigma_{ZG}\\
\Sigma_{ZG}^\top&\Sigma_{GG}
\end{pmatrix},
\]
where
\[
\Sigma_{ZG}
=
\begin{pmatrix}
R_1&R_2
\end{pmatrix}
\in\bbR^{\tilde d\times2\tilde d},
\qquad
\Sigma_{GG}
=
\begin{pmatrix}
I_{\tilde d}&R_1\\
R_1^\top&I_{\tilde d}
\end{pmatrix}
\in\bbR^{2\tilde d\times2\tilde d}.
\]
Here, \(\Sigma_{ZG}\) is the cross-covariance block between the current
vector and its two lags, whereas \(\Sigma_{GG}\) is the covariance
matrix of the lagged vectors. Then the Gaussian conditional expectation
is
\[
\bbE\lrp{\mZ_t\mid\mG_t}
=
\Sigma_{ZG}\Sigma_{GG}^{-1}\mG_t
=
\lrp{F_1,F_2}\mG_t,
\]
where the coefficient matrices are the solution to the multivariate
Yule--Walker equations,
\[
\lrp{F_1,F_2}
=
\lrp{R_1,R_2}
\begin{pmatrix}
I_{\tilde d}&R_1\\
R_1^\top&I_{\tilde d}
\end{pmatrix}^{-1}.
\]
Therefore,
\(\lrp{\mZ_t\mid\mZ_{t-1},\mZ_{t-2}}
\sim\calN_{\tilde d}
\lrp{F_1\mZ_{t-1}+F_2\mZ_{t-2},Q}\).

Define the innovation by
\(\uvec_t=\mZ_t-F_1\mZ_{t-1}-F_2\mZ_{t-2}\), whose covariance follows
the lag-zero Yule--Walker equation,
\[
\begin{aligned}
Q
&=
I_{\tilde d}
-
\begin{pmatrix}
R_1&R_2
\end{pmatrix}
\begin{pmatrix}
I_{\tilde d}&R_1\\
R_1^\top&I_{\tilde d}
\end{pmatrix}^{-1}
\begin{pmatrix}
R_1^\top\\
R_2^\top
\end{pmatrix} \\
&=
I_{\tilde d}-F_1R_1^\top-F_2R_2^\top.
\end{aligned}
\]
The Gaussian Markov property implies that \(\uvec_t\) is independent of
the entire history through time \(t-1\). Since \(F_1,F_2\), and \(Q\)
do not depend on \(t\), the innovations are identically distributed,
giving the exact stationary representation
\[
\mZ_t = F_1\mZ_{t-1} + F_2\mZ_{t-2} + \uvec_t,
\qquad
\uvec_t\sim\calN_{\tilde d}\lrp{\zerovec,Q}.
\]
Thus, the stationary VD-vine parameters \(M_1,M_2\) and the PVAR
parameters \(F_1,F_2,Q\) are two parameterizations of exactly the same
Gaussian dependence structure \(\mR\). The former describes serial
dependence through vector partial correlations, whereas the latter
describes it through linear regression coefficients and an innovation
covariance matrix.

For DGP2, we set marginal location and scale as
\(\muvec_t=\zerovec\) and \(D_t=I\) for every \(t\). The observed vector
is \(\mX_t=L\mZ_t\), where \(LL^\top=\Sigma\). Substituting
\(\mZ_t=L^{-1}\mX_t\) into the latent VAR representation gives
\[
\mX_t = A_1\mX_{t-1} + A_2\mX_{t-2} + \evec_t,
\qquad
\evec_t = L\uvec_t \sim \calN_{\tilde d}\lrp{\zerovec,Q_e},
\]
where \(A_1=LF_1L^{-1}\), \(A_2=LF_2L^{-1}\), and
\(Q_e=LQL^\top\).

\subsection{Additional Figures for Simulation Study}\label{sm:simfigs}
\begin{figure}[thbp]
    \centering
    \includegraphics[width=\linewidth]{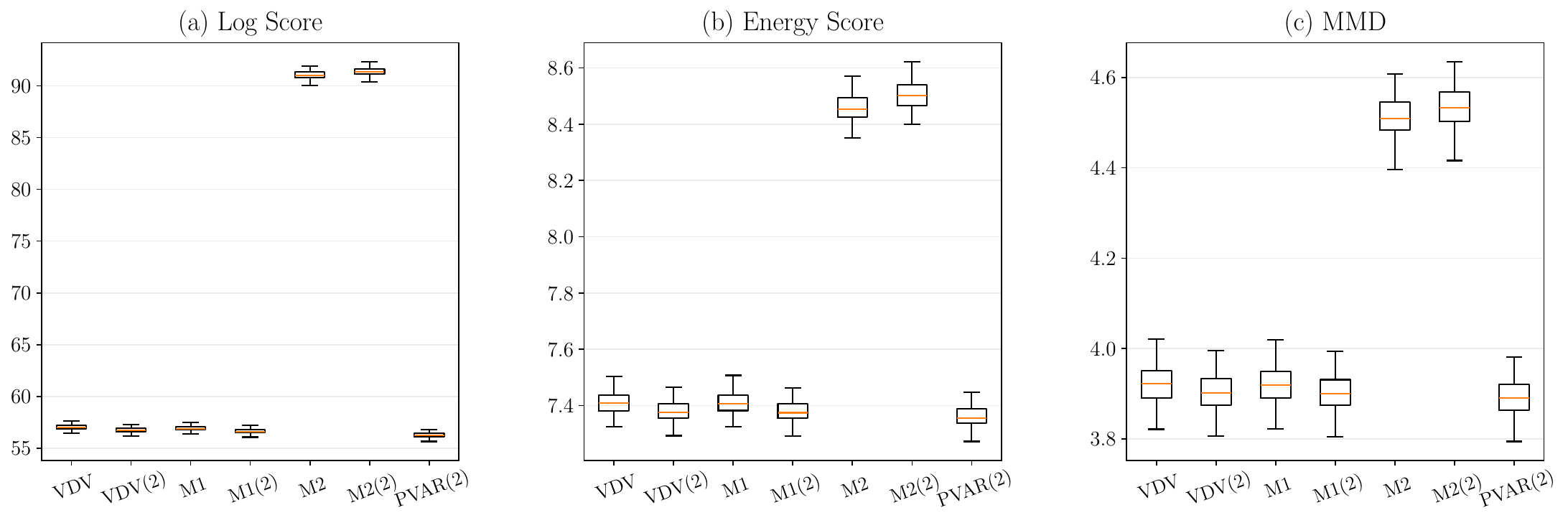}
    \caption{Boxplots of the log score, energy score,
    	and squared maximum mean discrepancy across 100 replications for DGP2, with lower values indicating greater accuracy.}
    \label{fig:simulation_boxplots_DGP2}
\end{figure}

\newpage

\section{Details on the Evaluation of Conditional Distributional Fit}\label{sm:evalfit}
To evaluate conditional distributional fit, we consider the negative logarithmic score (LS), energy score (ES), and maximum mean discrepancy (MMD) for the scoring function ${\cal S}$ in Sections~\ref{sec:sim} and~\ref{sec:unbalanced}. 
Let \(F\) be a \(d\)-variate conditional predictive distribution with density \(f\), let \(\yvec\) be the observed outcome, and let \(\widetilde{\yvec},\widetilde{\yvec}'\stackrel{\mathrm{iid}}{\sim}F\). The negative logarithmic score is \(\operatorname{LS}\lrp{F,\yvec}=-\log f\lrp{\yvec}\). If $d_0=\mbox{dim}(\yvec)$, then the dimension-normalized energy score is
\[
\operatorname{ES}\lrp{F,\yvec} =
{d_0}^{-1/2}\lrb{
\bbE\lrb{\lVert\widetilde{\yvec}-\yvec\rVert_2}
-\frac{1}{2}\bbE\lrb{\lVert\widetilde{\yvec}-\widetilde{\yvec}'\rVert_2} }.
\]
The squared MMD-based kernel score compares the predictive distribution \(F\) with the point mass at the observed outcome:
\[
\operatorname{MMD}\lrp{F,\yvec}=\bbE\lrb{k_r\lrp{\widetilde{\yvec},\widetilde{\yvec}'}}-2\bbE\lrb{k_r\lrp{\widetilde{\yvec},\yvec}}+k_r\lrp{\yvec,\yvec},
\]
where \(k_r\) is a Gaussian kernel with bandwidth selected by the median-distance heuristic using the observed target vectors in replication \(r\). The same kernel and bandwidth are used for all competing models within each replication.
ES and MMD are evaluated by Monte Carlo using \(1{,}000\) independent pairs of draws from each fitted conditional distribution.
Smaller values indicate better conditional distributional fit.

\newpage
\section{Benchmark PVAR with Missing Variables}\label{app:pvar}

We employ a Bayesian random-effects panel vector autoregression of order one (PVAR(1)) as a linear-Gaussian benchmark. Let \( \mZ_{it} = ( Z_{it,\mathrm{INC}},Z_{it,\mathrm{WLC}},Z_{it,\mathrm{DS}},Z_{it,\mathrm{PA}},Z_{it,\mathrm{NA}} )^\top \in \bbR^5 \) denote the complete vector for individual \(i\) at wave \(t\). All variables are expressed on the standardized scale. The complete process follows
\[
\mZ_{it} = \cvec + A\mZ_{i,t-1} + \mB_i + \varepsilonvec_{it},
\]
where \( \varepsilonvec_{it} \sim N_5( \zerovec,\Sigma_{\varepsilon} ) \) and \( \mB_i \sim N_5( \zerovec,\Sigma_b ) \). The vector \( \cvec \in \bbR^5 \) is the common intercept and \( A \in \bbR^{5\times5} \) is the transition matrix, whose diagonal elements describe persistence within each variable and off-diagonal elements describe cross-variable predictive effects. The covariance \( \Sigma_{\varepsilon} \) captures contemporaneous dependence in the innovations, while \( \Sigma_b \), assumed diagonal, captures time-invariant heterogeneity across individuals.

Let \( \mtildeZ_{it} \) denote the vector observed at wave \(t\), with realization \( \ztildevec_{it} \). At complete waves, \( \mtildeZ_{it} = \mZ_{it} \). In 2018, 2020, and 2022, WLC is unavailable and \( \mtildeZ_{it} = ( Z_{it,\mathrm{INC}},Z_{it,\mathrm{DS}},Z_{it,\mathrm{PA}},Z_{it,\mathrm{NA}} )^\top \in \bbR^4 \). Thus the tilde denotes the observed, and potentially incomplete, version of the complete vector. In particular, WLC is absent from \( \mtildeZ_{it} \) at these waves but remains a component of the dynamic process \( \mZ_{it} \).

The model is evaluated using a Kalman filter, which jointly propagates the complete vector \( \mZ_{it} \) and the time-invariant individual effect \( \mB_i \). Conditioning on the observed history \( \ztildevec_{i,1:t-1} \), the filter provides the one-step predictive distribution of the complete vector
\[
\mZ_{it} \mid \ztildevec_{i,1:t-1},\thetavec \sim N_5\lrp{ \muvec_{i,t|t-1},\Omega_{i,t|t-1} },
\]
where \( \muvec_{i,t|t-1} \) and \( \Omega_{i,t|t-1} \) are the predictive mean and covariance implied by \( \thetavec = \{ \cvec,A,\Sigma_{\varepsilon},\Sigma_b \} \) and the filtering distribution at wave \( t-1 \).

Let \( \widetilde{\muvec}_{i,t|t-1} \) and \( \widetilde{\Omega}_{i,t|t-1} \) denote the corresponding predictive mean and covariance of the observed vector \( \mtildeZ_{it} \). At complete waves, \( \widetilde{\muvec}_{i,t|t-1} = \muvec_{i,t|t-1} \) and \( \widetilde{\Omega}_{i,t|t-1} = \Omega_{i,t|t-1} \). When WLC is unavailable, \( \widetilde{\muvec}_{i,t|t-1} \) is obtained by removing its WLC entry from \( \muvec_{i,t|t-1} \), and \( \widetilde{\Omega}_{i,t|t-1} \) by removing the corresponding row and column from \( \Omega_{i,t|t-1} \). By Gaussian marginalization,
\[
\mtildeZ_{it} \mid \ztildevec_{i,1:t-1},\thetavec \sim N_{d_t}\lrp{ \widetilde{\muvec}_{i,t|t-1},\widetilde{\Omega}_{i,t|t-1} },
\]
where \( d_t = 5 \) at complete waves and \( d_t = 4 \) when WLC is unavailable. Hence missing WLC values do not need to be filled in before estimation; they are integrated out analytically from the complete Gaussian predictive distribution.

Conditioning on the fully observed first wave, the observed-data likelihood is
\[
f( \ztildevec;\thetavec ) = \prod_{i=1}^{n}\prod_{t=2}^{T}\phi_{d_t}\lrp{ \ztildevec_{it};\widetilde{\muvec}_{i,t|t-1},\widetilde{\Omega}_{i,t|t-1} }.
\]
Although WLC is unobserved at some waves, its predictive distribution continues to evolve through \( A \). The variables observed at the same wave also provide information about WLC through their predictive cross-covariances in \( \Omega_{i,t|t-1} \). The Kalman update therefore revises the distribution of the missing WLC component using the available INC, DS, PA, and NA observations and propagates this uncertainty to subsequent waves. When WLC is observed again, it re-enters the filtering recursion in the usual way.

After estimation, a missing WLC value can be recovered from the Kalman smoothing distribution \( p\lrp{ Z_{it,\mathrm{WLC}} \mid \ztildevec_{i,1:T},\thetavec } \), which uses observations both before and after the missing wave. Thus the Kalman filter marginalizes missing WLC values when evaluating the likelihood, while the smoother provides their conditional distributions when recovery of the missing values is required.

\newpage
\section{Further Information on the HILDA Panel Data Example}\label{sm:hildaextra}
\subsection{Data source and sample construction}
We used the Household, Income and Labour Dynamics in Australia (HILDA) Survey,
General Release 23, Version 2, covering Waves 1--23
\citep{NBTNMV_2024}. 
The analysis period was restricted to Waves 16--23 with individuals 
identified using the unique ID variable \texttt{xwaveid}. These waves correspond to 
time points $t=1,\ldots,T=8$ in the longitudinal model. 

The sample was formed in the following order. First, records treated as employed (\texttt{esempst} was neither $-10$ nor
$-1$) were retained. Observations with a valid response, from 1 to 7, to
\texttt{pawktdw} were retained, while preserving observations in Waves 18, 20, and 22, for
which the work-life-conflict was unobserved.
All negative survey codes for the component items and income were changed to
system missing.

For each multi-item scale, all observations for an individual were discarded if more than
one component was missing. A scale with no more than one missing component was
computed as the arithmetic mean of its available components. Overall life
satisfaction (\texttt{LS}) was a single item; observations with missing
\texttt{LS} were excluded during this preliminary cleaning step, although
\texttt{LS} was subsequently omitted from the analysis because it was discrete-valued on a small number of atoms. WLC, DS, PA and NA were treated as continuous for the purposes of analysis. The structural
work-life-conflict missingness in Waves 18, 20, and 22 was exempted from the
item-missingness exclusion.

Next, individuals with exactly one record in every wave from
16 through 23 were retained. Finally, we identified every individual whose gross
household wage and salary income (\texttt{INC}) equalled zero in at least one
of the eight waves and removed all records for that individual. There were only 58
such individuals. The resulting balanced panel contains $n=1,093$ individuals and
$nT=8,744$ individual-wave records. 
Balanced-panel status refers to participation in every wave; work-life conflict
remains structurally unavailable in waves 18, 20, and 22.

\subsection{Constructed variables}
The four psychosocial measures use the same constructs and component items as
\citet{Yuetal2025}, to the extent permitted by the selected waves and balanced-panel
design:

\begin{description}
	\setlength{\itemsep}{0.25em}
	\item[Work-life conflict (WLC).]
	The mean of the seven items:
	\texttt{pawkrap}, \texttt{pawklte}, \texttt{pawkmfh},
	\texttt{pawkfle}, \texttt{pawkle}, \texttt{pawkwc}, and
	\texttt{pawktdw}. Higher scores indicate greater conflict.
	
	\item[Domain satisfaction (DS).]
	The mean of eight 0--10 satisfaction items:
	\texttt{losateo}, \texttt{losatfs}, \texttt{losatft},
	\texttt{losathl}, \texttt{losatlc}, \texttt{losatnl},
	\texttt{losatsf}, and \texttt{losatyh}. Higher scores indicate greater
	satisfaction.
	
	\item[Positive affect (PA).]
	The mean of the four items:
	\texttt{gh9a}, \texttt{gh9d}, \texttt{gh9e}, and \texttt{gh9h}.
	The supplied construction leaves these items in their release coding, under
	which higher values indicate more frequent positive affect.
	
	\item[Reverse-coded Negative affect (NA).]
	The mean $\overline{\texttt{NA}}$ of five 1--6 SF-36 affect items
	(\texttt{gh9b}, \texttt{gh9c}, \texttt{gh9f}, \texttt{gh9g}, and
	\texttt{gh9i}) was first calculated and then transformed as
	$\texttt{NA}=7-\overline{\texttt{NA}}$. Higher values therefore
	indicate more frequent negative affect.
	
	\item[Income (INC).]
	Equal to the variable \texttt{hiwsfei}, which is the gross household wages and salary
	income for the financial year, measured in Australian dollars. Negative special
	codes were treated as missing, and the zero-income exclusion
	described above was applied. The logarithmic transformation was employed before the longitudinal models were applied.
\end{description}

This construction follows \citet{Yuetal2025} in its choice of work-life conflict,
domain satisfaction, positive affect, negative affect, and household income.
There are four deliberate differences. The present sample uses Waves 16--23
rather than Waves 1--21; requires the same individuals to appear in all eight
selected waves rather than including anyone who answered the work-family items
at least once; omits overall life satisfaction from the retained variables as this was a highly discrete variable; and
retains income in dollars for the summary table rather than converting it to
within-wave quintiles.

\subsection{Wave-specific descriptive statistics}
Table~\ref{tab:summary} reports statistics separately by wave, calculated
across the retained individuals. 

\begingroup
\footnotesize
\setlength{\tabcolsep}{4.5pt}
\begin{longtable}{@{}r l r r r r r r@{}}
	\caption{Summary statistics by survey wave}\label{tab:summary}\\
	\toprule
	Wave & Variable & $n$ & Mean & SD & Minimum & Maximum & Skewness \\
	\midrule
	\endfirsthead
	
	\multicolumn{8}{l}{\tablename\ \thetable\ (continued)}\\
	\toprule
	Wave & Variable & $n$ & Mean & SD & Minimum & Maximum & Skewness \\
	\midrule
	\endhead
	
	\midrule
	\multicolumn{8}{r}{Continued on next page}\\
	\endfoot
	
	\bottomrule
	\endlastfoot
	
	16 & Work-life conflict & 1,093 & 3.699 & 1.159 & 1.000 & 7.000 & 0.074 \\
	& Domain satisfaction & 1,093 & 7.254 & 1.039 & 3.625 & 10.000 & $-0.554$ \\
	& Positive affect & 1,093 & 2.905 & 0.922 & 1.000 & 6.000 & 0.673 \\
	& Negative affect (rev.) & 1,093 & 2.332 & 0.770 & 1.000 & 6.000 & 1.036 \\
	& Income (AUD) & 1,093 & 129,358 & 74,287 & 2,400 & 737,054 & 2.324 \\
	\addlinespace
	17 & Work-life conflict & 1,093 & 3.628 & 1.208 & 1.000 & 7.000 & 0.111 \\
	& Domain satisfaction & 1,093 & 7.287 & 1.065 & 3.000 & 10.000 & $-0.507$ \\
	& Positive affect & 1,093 & 2.956 & 0.954 & 1.000 & 6.000 & 0.610 \\
	& Negative affect (rev.) & 1,093 & 2.321 & 0.749 & 1.000 & 5.400 & 0.876 \\
	& Income (AUD) & 1,093 & 135,945 & 76,385 & 5,597 & 720,283 & 2.241 \\
	\addlinespace
	18 & Work-life conflict & 0 & -- & -- & -- & -- & -- \\
	& Domain satisfaction & 1,093 & 7.357 & 1.040 & 2.750 & 10.000 & $-0.522$ \\
	& Positive affect & 1,093 & 2.963 & 0.930 & 1.000 & 6.000 & 0.642 \\
	& Negative affect (rev.) & 1,093 & 2.326 & 0.762 & 1.000 & 5.600 & 1.094 \\
	& Income (AUD) & 1,093 & 145,308 & 81,882 & 3,000 & 645,709 & 1.926 \\
	\addlinespace
	19 & Work-life conflict & 1,093 & 3.636 & 1.160 & 1.000 & 7.000 & 0.098 \\
	& Domain satisfaction & 1,093 & 7.373 & 1.031 & 2.875 & 10.000 & $-0.467$ \\
	& Positive affect & 1,093 & 2.943 & 0.940 & 1.000 & 5.750 & 0.566 \\
	& Negative affect (rev.) & 1,093 & 2.325 & 0.768 & 1.000 & 5.400 & 0.968 \\
	& Income (AUD) & 1,093 & 152,555 & 89,483 & 5,500 & 778,086 & 2.279 \\
	\addlinespace
	20 & Work-life conflict & 0 & -- & -- & -- & -- & -- \\
	& Domain satisfaction & 1,093 & 7.539 & 0.960 & 3.875 & 10.000 & $-0.546$ \\
	& Positive affect & 1,093 & 3.042 & 0.950 & 1.000 & 5.750 & 0.519 \\
	& Negative affect (rev.) & 1,093 & 2.397 & 0.804 & 1.000 & 6.000 & 0.909 \\
	& Income (AUD) & 1,093 & 157,113 & 84,175 & 7,000 & 673,354 & 1.775 \\
	\addlinespace
	21 & Work-life conflict & 1,093 & 3.498 & 1.245 & 1.000 & 7.000 & 0.087 \\*
	& Domain satisfaction & 1,093 & 7.555 & 0.972 & 3.125 & 10.000 & $-0.603$ \\*
	& Positive affect & 1,093 & 3.073 & 0.955 & 1.000 & 6.000 & 0.439 \\*
	& Negative affect (rev.) & 1,093 & 2.405 & 0.807 & 1.000 & 6.000 & 0.933 \\*
	& Income (AUD) & 1,093 & 168,646 & 93,446 & 3,000 & 756,855 & 2.150 \\
	\addlinespace
	22 & Work-life conflict & 0 & -- & -- & -- & -- & -- \\
	& Domain satisfaction & 1,093 & 7.513 & 0.962 & 3.500 & 10.000 & $-0.380$ \\
	& Positive affect & 1,093 & 3.081 & 0.948 & 1.000 & 5.750 & 0.432 \\
	& Negative affect (rev.) & 1,093 & 2.427 & 0.831 & 1.000 & 5.800 & 0.950 \\
	& Income (AUD) & 1,093 & 183,078 & 101,144 & 2,000 & 780,917 & 1.881 \\
	\addlinespace
	23 & Work-life conflict & 1,093 & 3.402 & 1.224 & 1.000 & 7.000 & 0.147 \\
	& Domain satisfaction & 1,093 & 7.581 & 0.997 & 3.000 & 10.000 & $-0.542$ \\
	& Positive affect & 1,093 & 3.081 & 0.961 & 1.000 & 6.000 & 0.442 \\
	& Negative affect (rev.) & 1,093 & 2.427 & 0.801 & 1.000 & 5.400 & 0.978 \\
	& Income (AUD) & 1,093 & 196,723 & 106,969 & 24,000 & 955,681 & 2.286 \\
\end{longtable}
\endgroup

\noindent\textit{Note.} Statistics are computed separately within each wave
over individuals remaining after the zero-income exclusion.
Income is reported in nominal Australian dollars and rounded to the nearest
dollar in the table. Work-life conflict was not observed in Waves 18, 20, and
22; its zero counts represent structural non-administration, not sample
attrition.

\subsection{Additional Empirical Results}
\begin{table}[H]
	\centering
	\caption{HILDA 10-Fold CV Prediction Evaluation Scores}
	\label{tab:hilda_cv_scores}
	\small
	\resizebox{\textwidth}{!}{%
		\begin{tabular}{lccccc}
			\toprule
			\toprule
			& \multicolumn{5}{c}{VD-vine Model}\\ \cmidrule{2-6}
			Linking VC & VDV& M1 & M2 & M3 & M4 \\
			\midrule
			\multicolumn{6}{l}{\it Panel A: Negative Log Score} \\
			GVC & \textbf{106.08} (1.40) & 113.08 (1.31) & 116.32 (1.34) &  & 118.18 (1.08) \\
			FGMVC & 130.13 (1.22) & 135.40 (1.13) & 126.58 (1.11) & 141.45 (1.31) & 140.26 (1.30) \\
			Hybrid & 111.62 (1.07) & 118.49 (0.98) & 118.54 (1.26) &  & 118.09 (1.08) \\
			\midrule
			\multicolumn{6}{l}{\it Panel B: Normalized Energy Score} \\
			GVC & \textbf{11.82} (0.11) & 11.91 (0.11) & 12.10 (0.11) &  & 13.26 (0.08) \\
			FGMVC & 14.53 (0.11) & 14.57 (0.11) & 12.83 (0.11) & 16.21 (0.10) & 16.02 (0.10) \\
			Hybrid & 12.52 (0.09) & 12.63 (0.09) & 12.31 (0.11) &  & 13.24 (0.08) \\
			\midrule
			\multicolumn{6}{l}{\it Panel C: Normalized MMD Score} \\
			GVC & \textbf{10.66} (0.06) & 10.78 (0.05) & 10.90 (0.05) &  & 11.60 (0.03) \\
			FGMVC & 12.36 (0.02) & 12.43 (0.03) & 11.49 (0.05) & 13.23 (0.02) & 13.14 (0.01) \\
			Hybrid & 11.15 (0.03) & 11.27 (0.03) & 11.09 (0.05) &  & 11.59 (0.03) \\
			\bottomrule
			\bottomrule
		\end{tabular}
	}
	\caption*{Notes: Entries report the mean of the scores across all folds, with standard error in parentheses. Lower values are better for all three scores; the best value within each score block is bolded. VD-vine denotes the full vector D-vine model, while M1, M2, M3, and M4 correspond to nested sub-models. The model M3 has an independence vector vine copula, so that results are unaffected by linking vector copula choice.}
\end{table}

\subsection{Additional Figures}

\begin{figure}[H]
    \centering
    \caption{Ten-fold cross-validation boxplots of the MMD score for the HILDA application.}
    \includegraphics[width=\textwidth]{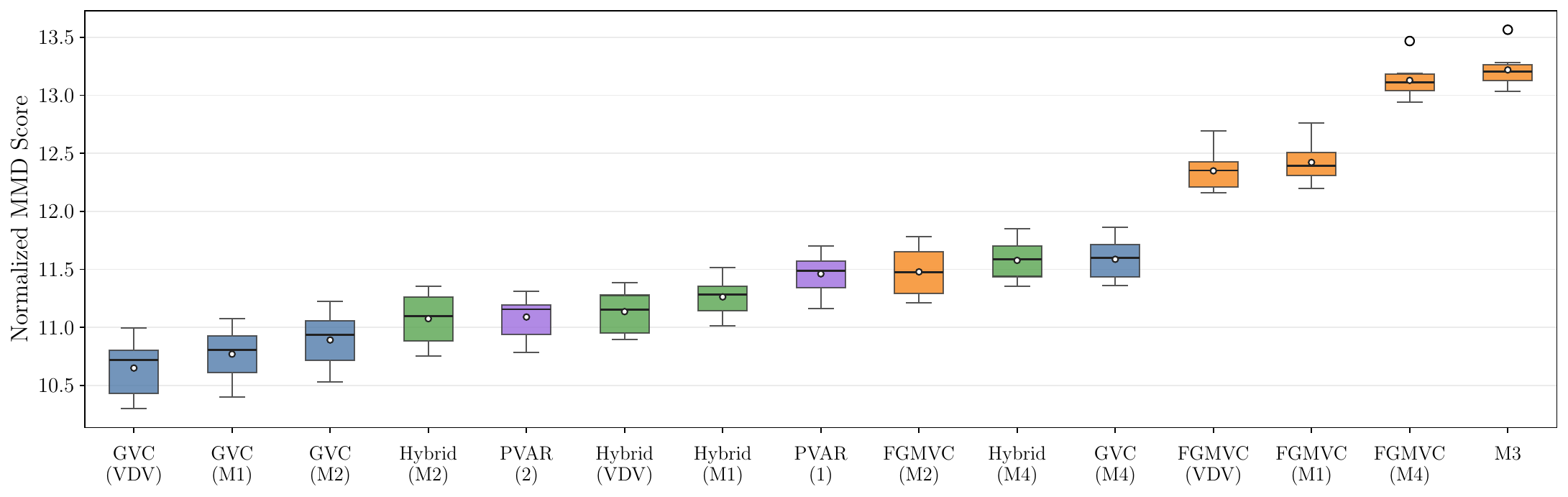}
    \label{fig:hilda_cv_mmd_boxplot}
\end{figure}

\begin{figure}[H]
    \centering
    \caption{Ten-fold cross-validation boxplots of LS for the HILDA application.}
    \includegraphics[width=\textwidth]{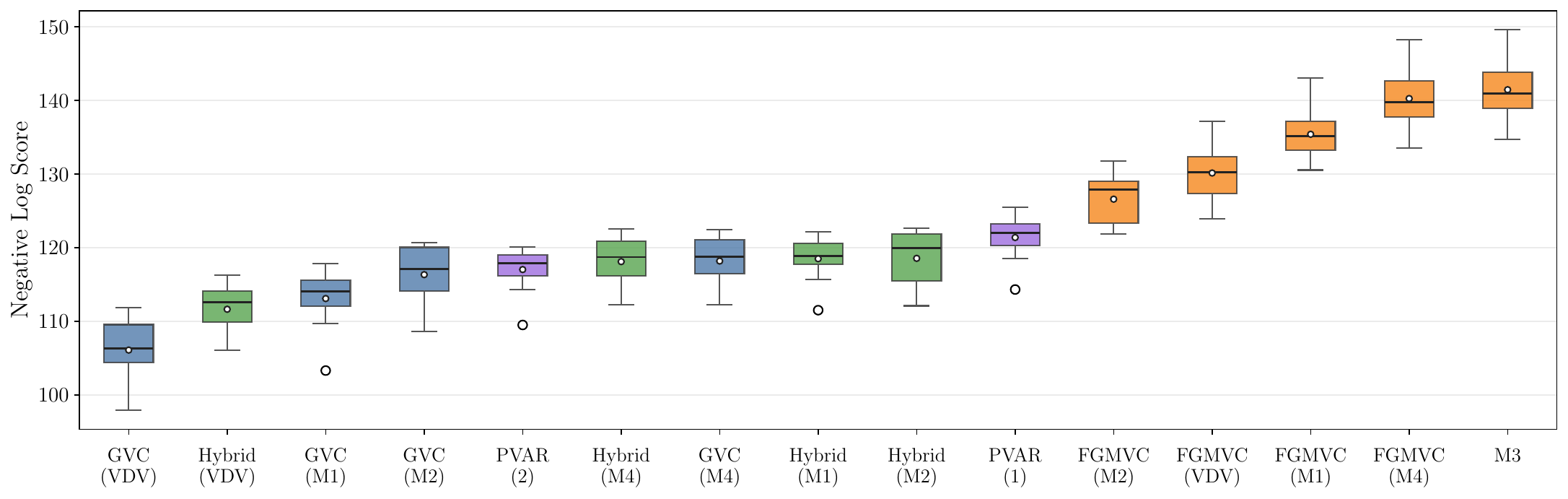}
    \label{fig:hilda_cv_logscore_boxplot}
\end{figure}

\begin{figure}[H]
    \centering
	\caption{Pairwise Spearman correlations from the VD-vine model. }
    \includegraphics[width=\textwidth]{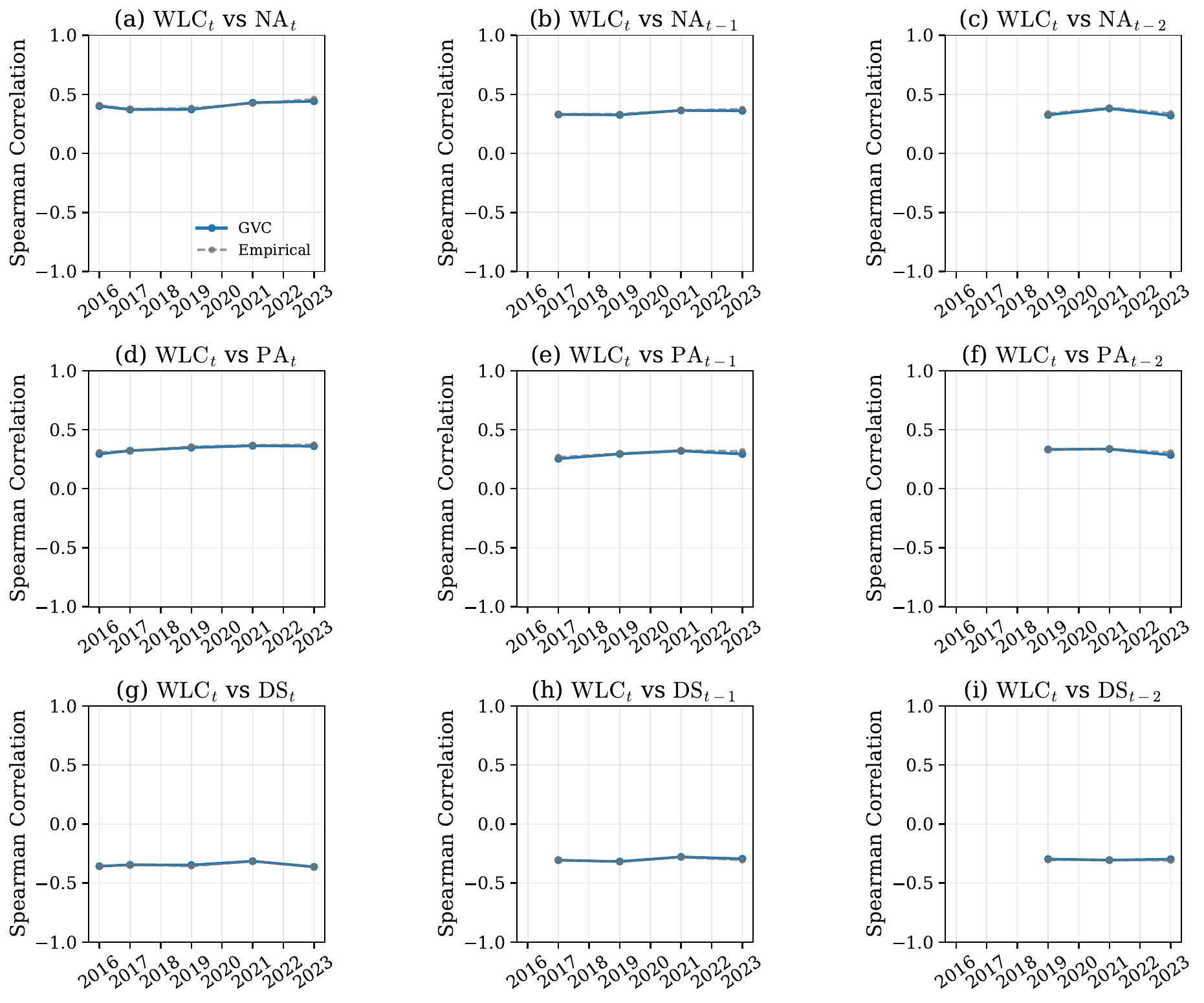}
    \caption*{Notes: Plots are Spearman correlations between WLC and the well-being measures NA, PA, and DS in the HILDA panel data. The horizontal axis gives the wave \(t\) from 2016 to 2023, and the vertical axis the Spearman correlation.}
    \label{fig:wlc_lagged_spearman}
\end{figure}

\newpage
\section{Extra Details on the Hyper-spherical Parametrization}\label{sm:priors}
Here we present some extra details on the hyper-spherical parameterization of the correlation matrix $\Sigma_j(\varphivec)$ in Section~\ref{sec:ltmaps}. Following~\cite{pinheiro1996unconstrained} and many
subsequent authors, we consider the upper triangular Cholesky factorization
\(\Sigma_j=A_j A_j^\top\) and parameterize $A_j$ in terms of $\varphivec_j$. Restrictions on $\varphivec_j$ give \(A_j\) a positive diagonal, so the resulting correlation matrix is positive definite and its Cholesky factor is unique, as below.

Writing the elements of the matrix $A_j=\{a_{j,kl}\}$, the corresponding lower-triangular hyperspherical Cholesky factor has entries
\begin{equation*}
	\begin{aligned}
		a_{j,kk} &=
		\begin{cases}
			1, & k=1,\\
			\displaystyle\prod_{m=1}^{k-1}\sin(\varphi_{j,km}), & k=2,\ldots,d_j,
		\end{cases}
		&
		a_{j,kl} &=
		\begin{cases}
			\cos(\varphi_{j,k1}), & l=1,\\
			\displaystyle
			\cos(\varphi_{j,kl})
			\prod_{m=1}^{l-1}\sin(\varphi_{j,km}), & l=2,\ldots,k-1,
		\end{cases}
	\end{aligned}
\end{equation*}
where the second expression applies for $k=2,\ldots,d_j$, the elements above the diagonal are zero, and $\varphi_{j,kl}\in(0,\pi)$.

\end{document}